\documentclass[sigconf]{acmart}

\def\fullFlag{1} 

 \AtBeginDocument{%
 \providecommand\BibTeX{{%
 \normalfont B\kern-0.5em{\scshape i\kern-0.25em b}\kern-.8em\TeX}}}

\copyrightyear{2023}
\acmYear{2023}
\setcopyright{acmlicensed}
\acmConference[CCS '23] {Proceedings of the 2023 ACM SIGSAC Conference on Computer and Communications Security}{November 26--30, 2023}{Copenhagen, Denmark.}
\acmBooktitle{Proceedings of the 2023 ACM SIGSAC Conference on Computer and Communications Security (CCS '23), November 26--30, 2023, Copenhagen, Denmark}

\acmPrice{15.00}
\acmDOI{10.1145/3576915.3616592}
\acmISBN{979-8-4007-0050-7/23/11}

\microtypecontext{spacing=nonfrench}

\usepackage{todonotes}

\usepackage{amsmath}
\usepackage{amsthm} 
\usepackage{amsfonts}
\usepackage{dsfont}
\usepackage{xspace}
\usepackage{microtype}
\usepackage{graphicx}
\usepackage[lined,ruled]{algorithm2e}
\usepackage{algpseudocode}
\usepackage{circledsteps} 
\usepackage{bbm}
\usepackage{hyperref}
\usepackage{multirow}
\usepackage{diagbox}

\newcommand{\cf}{{\it cf.}\xspace}
\newcommand{\eg}{{\it e.g.}\xspace}
\newcommand{\etal}{{\it et~al.}\xspace}

\newcommand{\ie}{{\it i.e.}\xspace}

\newcommand{\B}{\mathcal{B}} 
\renewcommand{\L}{\mathcal{L}} 
\newcommand{\M}{\mathcal{M}} 
\newcommand{\MN}{\mathcal{MN}} 
\newcommand{\N}{\mathcal{N}} 
\renewcommand{\O}{\mathcal{O}} 
\newcommand{\V}{\mathcal{V}} 

\newcommand{\Att}{\mathsf{Att}}
\newcommand{\FFN}{\mathsf{FFN}}
\newcommand{\MHA}{\mathsf{MHA}}
\newcommand{\ReLU}{\mathsf{ReLU}}

\newcommand{\LN}{\mathsf{LN}}

\newcommand{\erf}{\mathsf{erf}}
\newcommand{\softmax}{\mathsf{softmax}}

\newcommand{\lip}{\mathit{Lip}}

\newcommand{\vect}{\mathit{vec}}
\newcommand{\D}{\mathcal{X}}

\newcommand{\aux}{\mathrm{aux}}
\newcommand{\MLM}{\mathrm{MLM}} 

\DeclareMathOperator{\Tr}{Tr}

\newtheorem{definition}{\bf Definition}
\newtheorem{theorem}{\bf Theorem}
\newtheorem{proposition}{\bf Proposition}
\newtheorem{lemma}{\bf Lemma}

\begin{document}

\title{DP-Forward: Fine-tuning and Inference on Language Models with Differential Privacy in Forward Pass}

\settopmatter{authorsperrow=3}
\author{Minxin Du}
\affiliation{%
	\institution{The Chinese University of Hong Kong}
	\country{}
 }
\email{dm018@ie.cuhk.edu.hk}

\author{Xiang Yue}
\authornote{The first two authors contributed equally to this work.}
\affiliation{%
	\institution{The Ohio State University}
	\country{}
 }
\email{yue.149@osu.edu}

\author{Sherman S. M. Chow}
\authornote{Corresponding author is from Dept. of Information Engineering, CUHK, Hong Kong.}
\affiliation{%
	\institution{The Chinese University of Hong Kong}
	\country{}
 }
\email{smchow@ie.cuhk.edu.hk}

\author{Tianhao Wang}
\affiliation{%
	\institution{University of Virginia}
	\country{}
 }
\email{tianhao@virginia.edu}

\author{Chenyu Huang}
\affiliation{%
	\institution{Independent}
	\country{}
 }
\email{hcyray@gmail.com}

\author{Huan Sun}
\affiliation{%
	\institution{The Ohio State University}
	\country{}
 }
\email{sun.397@osu.edu}

\renewcommand{\shortauthors}{Du, et al.}

\begin{abstract}
Differentially private stochastic gradient descent (DP-SGD) adds noise to gradients in back-propagation, safeguarding training data from privacy leakage, particularly membership inference.
It fails to cover (inference-time) threats like embedding inversion and sensitive attribute inference.
It is also costly in storage and computation when used to fine-tune large pre-trained language models (LMs).

We propose DP-Forward, which directly perturbs embedding \emph{matrices} in the forward pass of LMs.
It satisfies stringent \emph{local} DP requirements for training and inference data.
To instantiate it using the smallest matrix-valued noise, we devise an analytic matrix Gaussian~mechanism (aMGM) by drawing possibly non-i.i.d. noise from a \emph{matrix} Gaussian distribution.
We then investigate perturbing outputs from different hidden (sub-)layers of LMs with aMGM noises.
Its utility on three typical tasks almost hits the non-private baseline and outperforms DP-SGD by up to~$7.7$pp at a moderate privacy~level.
It saves $3\times$ time and memory costs compared to DP-SGD with the latest high-speed library.
It also reduces the average success rates of embedding inversion and sensitive attribute inference by up to $88$pp and $41$pp, respectively, whereas DP-SGD fails.
\end{abstract}



\maketitle

\section{Introduction}
The deep learning architecture of transformer~\cite{nips/VaswaniSPUJGKP17} is now gaining popularity in computer vision and has been widely utilized in natural language processing (NLP).
Transformer-based language models (LMs), such as BERT~\cite{naacl/DevlinCLT19} and GPT~\cite{report/RadfordNSS18,report/radford2019language}, have remarkably achieved state-of-the-art performance in almost every NLP~task.
They are first pre-trained on massive (public) self-labeled corpora and then fine-tuned for various tasks using much smaller, potentially private corpora.
It avoids training from scratch and the possible shortage of task-specific corpora while earning versatility.

Training data contributing to the improved utility of fine-tuned LMs can be sensitive.
LMs can (unintentionally) memorize them~\cite{uss/Carlini0EKS19} and become vulnerable to membership inference attacks (MIAs)~\cite{sp/ShokriSSS17} that identify whether an example is in the training set.
Worse still, verbatim training text (\eg, SSNs) can be extracted via only black-box access to GPT-2~\cite{uss/CarliniTWJHLRBS21}.
It is also possible to recover personal health information (\eg, patient-condition pairs) from BERT trained over a clinical corpus~\cite{naacl/LehmanJPGW21} based on the extraction attack~\cite{uss/CarliniTWJHLRBS21}.

Differential privacy (DP)~\cite{tcc/DworkMNS06} has emerged as the \emph{de facto} privacy standard for 
protecting individual privacy.
To thwart MIAs on individuals' \emph{training} data, 
DP stochastic gradient descent (DP-SGD)~\cite{ccs/AbadiCGMMT016} can be used.
It clips the gradients of each example in a batch and adds random Gaussian noise to the aggregated gradient.
It is more general than earlier attempts~\cite{nips/ChaudhuriM08, jmlr/ChaudhuriMS11} that focus on convex problems and has been implemented in modern ML frameworks, such as \href{https://opacus.ai}{PyTorch} and \href{https://github.com/tensorflow/privacy}{TensorFlow}.
One can apply it to fine-tune LM-based NLP pipelines while ensuring \emph{example-level} privacy, assuming each individual contributes an example, typically a sequence-label~pair.

Unfortunately, DP-SGD often uses a \emph{trusted} party to curate users' sensitive training data.
Although it can be done distributively~\cite{ccs/BonawitzIKMMPRS17,iclr/McMahanRT018} via secure aggregation~\cite{ccs/ChaseC09} with extra costs and trust assumptions, it offers \emph{central} DP (CDP) at its core.\footnote{Distributed DP-SGD adds local noise too small to achieve LDP.
But it is protected by secret sharing.
When \emph{all} shares are aggregated,
they cancel out each other, assuming an honest majority.
It thus faces a ``synchronization'' issue begging for identification and recovery mechanisms with computation and communication overheads~\cite{ccs/BonawitzIKMMPRS17}.}
Instantiating \emph{per-example} gradients as large as \emph{entire} pipelines (\eg, ${>}110$M parameters for BERT-Base) is obliviously costly.
Moreover, maintaining the utility of pipelines trained by the noisy aggregated one is tricky due to the dimensional ``curse.''
A recent study~\cite[Table~$4$]{icml/YuZC0L21} shows that
the average accuracy in fine-tuning LMs for four NLP tasks at moderate privacy is $65.7\%$ (vs. $91.8\%$ without DP).
Finally, the inference-time embeddings are not perturbed by the noise added during training, leaving \emph{inference queries} vulnerable to various recovery attacks~\cite{ccs/SongR20, sp/PanZJY20}, ranging from sensitive attributes (\eg, authorship) to raw text.

\subsection{Natural Privacy via Perturbing Embeddings} 
We propose DP-Forward, a radically different approach that perturbs \emph{forward-pass} signals:
Users can locally inject noise into the \emph{embeddings} of (labeled) sequences before sharing them for training, in contrast to perturbing gradients in back-propagation (possibly by an untrusted party).
It is meant for provable \emph{local} DP (LDP) guarantees, thus protecting against stronger adversaries than DP-SGD.

Our approach also naturally fits the federated learning (FL) setting that does not gather users' data but with substantial differences -- FL typically shares \emph{noiseless} local model updates.
Note that any subsequent computation (\eg, gradient computation) on noisy embeddings incurs no extra privacy loss due to the free post-processing of LDP.
One might force DP-SGD to offer LDP by adding ``enough'' noise to the \emph{orders-of-magnitude larger} per-example gradient from a user, but it may yield unusable models at a similar privacy level.

DP-Forward also extends its applicability to \emph{inference} 
via adding noise to users' test-sequence embeddings, ensuring LDP as in training.
As a ``side'' benefit, it can effectively mitigate emerging embedding-based privacy risks~\cite{sp/PanZJY20,ccs/SongR20} beyond MIAs.

It is evident that the design goals of DP-Forward naturally align in tandem with our overarching objectives: LDP (vs. CDP), more direct protection of raw data (vs. gradients) against new threats~\cite{sp/PanZJY20,ccs/SongR20}, and can be as efficient as regular non-private training (allowing batch processing of noisy embeddings).
The foundation supporting these desiderata, unfortunately, was unavailable.
A dedicated mechanism to perturb the forward-pass signals is indispensable.

Specifically, we need to derive noises for embeddings of training/inference text sequences
obtained through the forward pass of LM-based pipelines as a real- and matrix-valued function.
One might adopt the classical Gaussian mechanism (GM)~\cite{fttcs/DworkR14} to add i.i.d. noise drawn from a univariate Gaussian distribution.
Yet, GM calibrates its noise variance based solely on a \emph{sufficient condition} for DP, and its variance formula is not applicable to a low privacy regime~\cite{icml/BalleW18}.
Another candidate is the matrix-variate Gaussian (MVG) mechanism~\cite{ccs/ChanyaswadDPM18}, tailored for matrix-valued data:
It exploits possibly \emph{non-i.i.d.} noise from a \emph{matrix} Gaussian distribution to perturb more important rows/columns less.
Although it may show better utility over GM~\cite{ccs/ChanyaswadDPM18}, it is still sub-optimal due to the sufficient condition.

To optimize MVG, we propose an analytic matrix Gaussian mechanism (aMGM) by integrating a \emph{necessary and sufficient} condition from the analytic GM (aGM)~\cite{icml/BalleW18} for non-i.i.d. noise calibration.
Our challenge lies in manipulating the two covariance matrices instead of a single variance.
We deduce a constraint only on
the \emph{two smallest singular values} (Section~\ref{sect:agm-aMGM}), indicating that i.i.d. noise (as in aGM) may already be 
optimal for general applications like DP-Forward.\footnote{
With extra assumptions, dedicated allocation of other singular values by optimizing/maximizing utility functions specific to applications could help.
}

A transformer-based pipeline contains an input embedding layer, encoders, and task layers.
All these layers prominently manipulate embeddings of text inputs in training and subsequent inference.
We investigate adding aMGM noise to 
embeddings output by any hidden (sub-)layer
before task layers (Figure~\ref{fig:encoder}).
To ensure \emph{sequence-level} LDP, we need to estimate the $L_2$-sensitivity~\cite{fttcs/DworkR14} of ``pre-noise'' functions for \emph{any} two sequences.
It is non-trivial since the functions can include different (sub-)layers that may not even be Lipschitz~\cite{icml/KimPM21}.
Our strategy is to normalize the function outputs to have a fixed Frobenius (or $L_2$) norm, similar to gradient clipping~\cite{ccs/AbadiCGMMT016}.
It works especially well for deeper sub-layers, achieving comparable task accuracy to the non-private baseline (Section~\ref{sect:exp}).
\ifnum\fullFlag=1
For the first few (sub-)layers, we also make two specializations in relaxing LDP to the token level, elaborated in Appendix~\ref{sect:token-level}, to improve accuracy.
\else
For the first few (sub-)layers, we also make two specializations in relaxing LDP to the token level to improve accuracy (in our full version~\cite{corr/DuYCWHS23}).
\fi

\begin{figure*}[!t]
\centering
\includegraphics[width=0.8\linewidth]{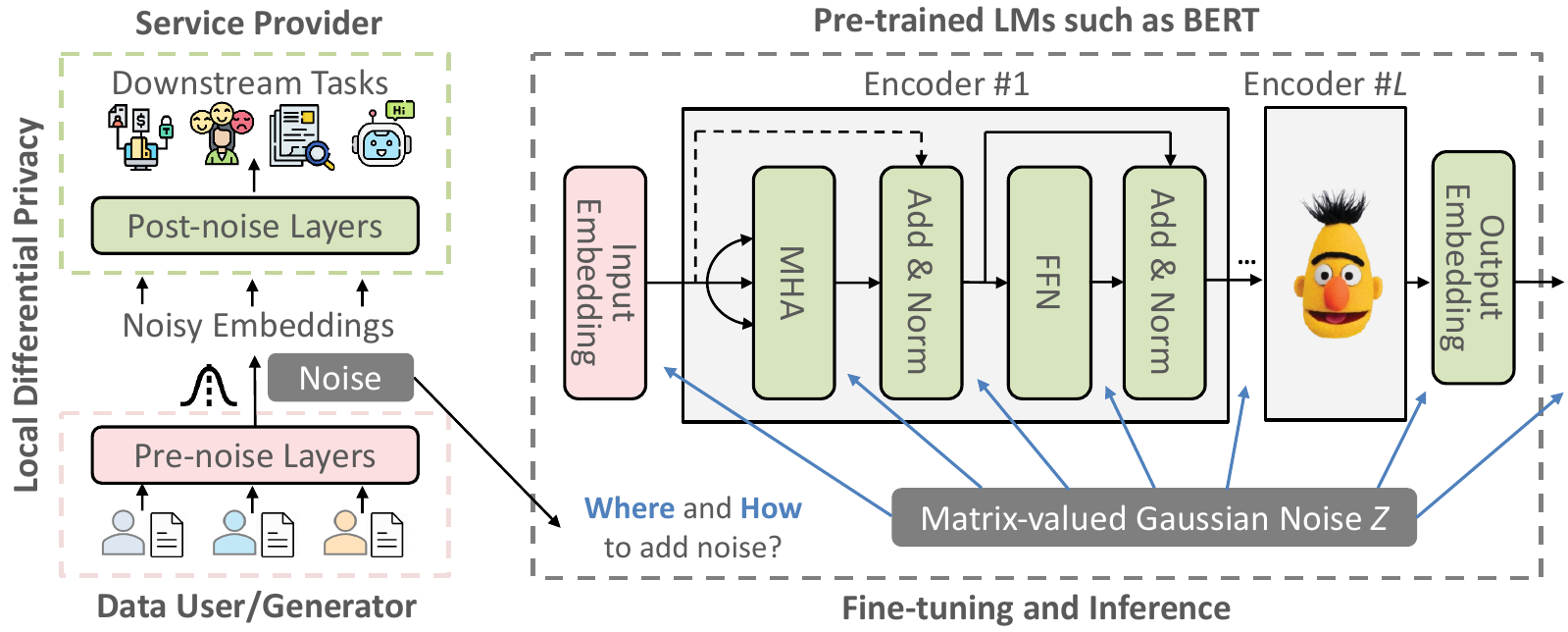}
\vspace{-5pt}
\caption{A typical NLP pipeline built atop a pre-trained LM such as BERT with our matrix-valued Gaussian noise layer}
\vspace{-8pt}
\label{fig:encoder}
\end{figure*}

\subsection{Our Contributions}
\label{subsec:contribution}
Motivated by prevailing privacy concerns in LM fine-tuning and inference and inherent shortcomings of \mbox{DP-SGD}, we initiate a formal study of an intuitive but rarely studied approach and explore its integration with a transformer-based NLP pipeline.
Specifically:

\smallskip
\noindent 1) We propose DP-Forward fine-tuning, which perturbs the forward-pass embeddings of \emph{every} user's (labeled) sequence.
It offers more direct protection than DP-SGD perturbing aggregated gradients.
Its provable guarantee (Theorem~\ref{theo:seq-ldp}) is a new sequence-level LDP notion (SeqLDP, Definition~\ref{def:seqldp}), with the more stringent $(\epsilon, \delta)$-LDP guarantee to hold w.r.t. only sequences.
Moreover, DP-Forward can naturally extend to inference, ensuring the standard LDP (Theorem~\ref{theo:inference_ldp}) for test sequences without labels, whereas DP-SGD cannot.

\smallskip
\noindent
2) To instantiate an optimal output perturbation mechanism for DP-Forward, we propose aMGM, owning independent interests for any matrix-valued function.
By exploiting a necessary and sufficient~DP condition from aGM~\cite{icml/BalleW18}, it can draw possibly non-i.i.d. noise from a \emph{matrix} Gaussian distribution like MVG~\cite{ccs/ChanyaswadDPM18} while producing orders-of-magnitude smaller noise for high-dimensional data (Section~\ref{sect:different_mech}).

\smallskip
\noindent
3) We conduct experiments\footnote{Our code is available at \url{https://github.com/xiangyue9607/DP-Forward}.} on three typical NLP tasks in Section~\ref{sect:exp},
showing how crucial hyperparameters (\eg, the sequence length) impact task accuracy.
To fairly compare with DP-SGD on privacy-vs.-utility: 
i) We perturb labels by the randomized response~\cite{jasa/Warner65} such that DP-Forward fine-tuning offers the standard LDP for sequence-label pairs (Theorem~\ref{theo:LDP_fine-tuning}).
ii) We ``translate'' DP-Forward with standard LDP to (example-level) CDP (as offered by DP-SGD) via shuffling~\cite{soda/ErlingssonFMRTT19}.
Our accuracy gain (for deep-layer DP-Forward instantiations) is up to $7.7$ percentage points (pp), compared to DP-SGD or its recent improvements~\cite{icml/YuZC0L21, iclr/YuNBGIKKL22} (reviewed in Section~\ref{sect:related_dp-sgd}), at a similar privacy level.
Efficiency-wise, DP-SGD incurs ${>}3\times$ time and GPU-memory costs even with the latest Opacus library~\cite{opacus}.

\smallskip
\noindent
4) We evaluate three classes of privacy threats.
Like DP-SGD, 
DP-Forward (including the two token-level designs 
\ifnum\fullFlag=1
in Appendix~\ref{apdx:exp})
\else
in our full version)
\fi
can effectively defend against sequence-level MIAs, but only DP-Forward can thwart the two threats on (inference-time) embeddings.
Specifically, Section~\ref{sect:attacks} shows that DP-SGD \emph{totally} fails in two embedding inversion attacks, while DP-Forward remarkably reduces their success rates by up to $88$pp.
For a neural-network-based attribute inference attack, DP-SGD reduces its success rates by only $15$pp on average, while DP-Forward achieves ${\sim}41$pp reduction, making the attack predict like assigning all labels to the majority class.

\smallskip
In short, DP-Forward is a better alternative to DP-SGD in training (and testing) deep-learning models, \eg, gigantic LM-based ones.

\section{Preliminaries and Notations}
\label{sect:pre}

\subsection{Transformer Encoders in BERT}
\label{sect:pre_BERT}
Modern transformer-based LMs, including BERT~\cite{naacl/DevlinCLT19} and GPT~\cite{report/RadfordNSS18}, are first pre-trained on enormous unannotated (public) corpora to learn contextualized text representations.
Later, they can be fine-tuned for various downstream NLP tasks (\eg, sentiment analysis, question answering) using much smaller, task-specific datasets.

We consider BERT (Figure~\ref{fig:encoder}), which comprises a stack of $L$ identical layers (\ie, bidirectional transformer encoders~\cite{nips/VaswaniSPUJGKP17}).
Each layer has two sub-layers:
the dot-product \emph{multi-head attention} (MHA)~\cite{nips/VaswaniSPUJGKP17} with $h$ heads and a \emph{feed-forward network} (FFN).
Each sub-layer has an extra residual connection, followed by layer normalization~\cite{corr/BaKH16}.

Let $X = \langle x_i\rangle^n_{i=1}$ be an input sequence of $n$ tokens (\eg, characters, words, sub-words, q-grams), where $x_i$ is from a vocabulary $\V$.
The input embedding layer first maps each $x_i$ to its representation in $\mathbb{R}^d$, which is the sum of the token, segment, and position embeddings.
We re-use $X$ to represent the hidden \emph{embedding matrix} in $\mathbb{R}^{n \times d}$.
For each of $h$ attentions $\Att_i$ in the MHA layer,
we derive the query, key, and value matrices $Q, K, V \in \mathbb{R}^{n \times d/h}$ ($h$ divides $d$) by multiplying~$X$ with head-specific weights $W^Q, W^K, W^V \in \mathbb{R}^{d \times d/h}$.
Its output is
$$\Att_i(Q,K,V) = \softmax (\frac{Q K^\top}{\sqrt{d/h}})V, \forall i \in [1,h].$$

The input to $\softmax(\cdot)$ is an $n \times n$ matrix of pairwise dot products.
Finally, MHA concatenates (denoted by $||$) all the head outputs into a matrix in $\mathbb{R}^{n \times d}$, right multiplied by a projection matrix $W^O \in \mathbb{R}^{d\times d}$:
$$\MHA(X) = [\Att_1 || \cdots || \Att_h] W^O.$$

FFN is composed of two linear mappings with a ReLU activation in between.
It separately and identically operates on each $x_{i \in [1, n]}$,
$$\FFN(x_i) = \ReLU(0,x_i W_1+b_1)W_2+b_2,$$
where $W_1$, $W_2$, $b_1$, and $b_2$ are trainable matrix/vector-valued parameters.
Its output on $X$ is 
$\FFN(X) = [\FFN(x_1)^\top || \cdots || \FFN(x_n)^\top]$.
The residual connection for sub-layers is $X + \MHA(X) / \FFN (X)$.
The layer normalization $\LN(x_i)$
normalizes all $x_i$ entries to have zero mean and unit variance using an extra scale-then-shift step.

At the output of the final encoder, the hidden embedding matrix is reduced to a sequence feature in $\mathbb{R}^{1\times d}$.
Standard reduction methods include mean pooling~\cite{emnlp/ReimersG19} (computing $\sum^n_{i=1} x_i / n$) or taking the last embedding of a special token \texttt{[CLS]} for classification~\cite{naacl/DevlinCLT19}.

The pre-training of BERT is based on two self-supervised tasks: \emph{masked language model} (MLM) and next sentence prediction~\cite{naacl/DevlinCLT19}.
We adopt MLM: It randomly masks out some tokens, indexed by $\mathcal{I}$, in an input sequence $X$.
The objective is to predict those masked tokens using their context by minimizing the cross-entropy loss
\begin{align}
\label{eq:mlm}
    L_{\MLM} = - \sum_{i \in \mathcal{I}} \log \Pr[x_i|\hat{X};\theta], \textrm{~with~} \hat{X} = X \setminus \{x_i | i \in \mathcal{I} \},
\end{align} 
where $\theta$ denotes all the parameters of BERT transformer encoders.

\subsection{(Local) Differential Privacy}
DP~\cite{tcc/DworkMNS06} is a rigorous, quantifiable privacy notion.
It has two popular models, \emph{central} and \emph{local}.
In central DP, a \emph{trusted} data curator accesses the set $\D$ of all individuals' raw data and processes $\D$ by a randomized mechanism $\M$ with some random noise.
Formally:

\begin{definition}[Central DP]
\label{def:cdp}
For privacy parameters $\epsilon \geq 0$ and $0 \leq \delta \leq 1$, $\M$ fulfills $(\epsilon,\delta)$-DP if, for all neighboring datasets $\D$ and~$\D'$ (denoted by $\D \simeq \D'$) and any 
subset $\O$ of the outputs of~$\M$, 
\begin{align*}
    \Pr[\M(\D) \in \O] \leq e^\epsilon \Pr[\M(\D') \in \O] + \delta.
\end{align*}
We call it $\epsilon$-DP or pure DP when $\delta = 0$.
\end{definition}

The neighboring notion is application-dependent (to be discussed in Section~\ref{sect:new_notion}).
Typically, it involves the ``replace-one'' relation: $\D'$ can be obtained from $\D$ by replacing a single individual's data point (\eg, a sequence-label pair).
CDP offers plausible deniability to any individual in a dataset.
In contrast, local DP (LDP)~\cite{focs/KasiviswanathanLNRS08} removes the trusted curator, allowing individuals to locally perturb their data using $\M$ before being sent to an untrusted aggregator for analytics.

\begin{definition}[Local DP]
\label{def:ldp}
For $\epsilon \geq 0, 0 \leq \delta \leq 1$, $\M$ is $(\epsilon,\delta)$-LDP if, for any two inputs $X, X'$ and any possible output subset $\O$ of $\M$, 
\begin{align*}
    \Pr[\M(X) \in \O] \leq e^\epsilon \Pr[\M(X') \in \O] + \delta.
\end{align*}
Similarly, we call it $\epsilon$-LDP when $\delta= 0$.
\end{definition}

\paragraph{Privacy Loss Random Variable (PLRV)}
For a specific pair of inputs $\D \simeq \D'$, the \emph{privacy loss} (or the ``actual $\epsilon$ value'')~\cite{icml/BalleW18} incurred by observing an output $O$ is the log-ratio of two probabilities:
$$\L_{\M,\D,\D'}(O) = \ln{\frac{\Pr[\M(\D) = O]}{\Pr[\M(\D') = O]}}.$$
When $O$ varies according to $\M(\D)$, we get the PLRV $\L_{\M, \D,\D'}$.
A helpful way to work with DP is to analyze tail
bounds on PLRVs~\cite{fttcs/DworkR14}, which we utilize to build our proposed mechanism in Section~\ref{sect:agm-aMGM}.

DP has two desirable properties: \emph{free post-processing} and \emph{composability}.
The former means that further computations on the outputs of an $(\epsilon,\delta)$-DP mechanism incur no extra privacy loss.
The latter allows us to build more complicated mechanisms atop simpler ones: sequentially (and adaptively) running an $(\epsilon,\delta)$-DP mechanism for $k$ times on the same input is at least $(k\epsilon, k\delta)$-DP.
The two properties also hold for LDP when considering a dataset has only one row.

An output perturbation mechanism $\M$ for a matrix-valued function $f: \D \rightarrow \mathbb{R}^{n \times d}$ is given by computing $f$ on the inputs and then adding random noise drawn from a random variable to its outputs.

\smallskip
\noindent \textbf{Gaussian Mechanism (GM).}
For $(\epsilon,\delta)$-DP, a typical instance of $\M$ is the classical GM~\cite{fttcs/DworkR14}, which adds noise $Z \in \mathbb{R}^{n \times d}$ with each entry i.i.d. drawn from a univariate Gaussian distribution $\N(0,\sigma^2)$.
The variance $\sigma^2 = 2 \ln (1.25/\delta)S^2_2(f) / \epsilon^2$ with the $L_2$-sensitivity:
$$S_2(f)= \sup_{\D \simeq \D'} ||f(\D)-f(\D')||_F,$$
where $||\cdot||_F$ denotes the matrix Frobenius norm~\cite{cu/HJ2012}.

Table~\ref{tbl:acronyms} summarizes the acronyms throughout this work.

\begin{table}[!t]
    \centering
    \caption{Acronyms (newly proposed ones are marked with $^*$)}
    \label{tbl:acronyms}
    \resizebox{\linewidth}{!}{
    \begin{tabular}{l||l}
        \toprule
        NLP & Natural Language Processing \\
        LM & Language Model \\
        BERT & Bidirectional Encoder Representations from Transformers \\
        MLM & Masked Language Modeling \\
        MHA & Multi-Head Attention \\
        FFN & Feed-Forward Network \\
        \midrule
        MIA & Membership Inference Attack \\
        DP-SGD & Differentially Private Stochastic Gradient Descent \\
        PLRV & Privacy Loss Random Variable \\
        \midrule
        (C)DP & (Central) Differential Privacy\\
        LDP & Local Differential Privacy \\
        Seq(L)DP$^*$ & Sequence (Local) Differential Privacy \\
        \midrule
        GM & Gaussian Mechanism \\
        RR & Randomized Response \\
        MVG & Matrix-Variate Gaussian (Mechanism) \\
        aGM & Analytic Gaussian Mechanism \\
        aMGM$^*$ & Analytic Matrix Gaussian Mechanism \\
        \bottomrule
    \end{tabular}
    }
\end{table}

\section{DP-Forward}\label{sect:dp-forward}
We study BERT-based pipelines as an example due to their superior performance in classification tasks.
DP-Forward can be readily applied to other (transformer-based) NLP or computer vision models that involve matrix-valued computation during the forward pass.

Suppose each user holds a sequence-label pair $(X, y)$ or only $X$ for fine-tuning or testing a pipeline at an \emph{untrusted} service provider.
Sharing redacted $X$ (with common PII removed) or its feature, a non-human-readable real-valued embedding matrix, is leaky~\cite{jots/Sweeney15,ccs/SongR20,sp/PanZJY20}.

For DP-Forward training, users perturb their embedding matrices locally to ensure (new notions of) LDP before being shared, and they should also perturb the corresponding labels if deemed sensitive (Section~\ref{sect:cmp_dpsgd}).
We explore different options for splitting pipelines into \emph{pre-noise} functions $f(\cdot)$ and \emph{post-noise} processing $p(\cdot)$ in Section~\ref{sect:splitting}:
Users can access $f(\cdot)$ to derive embedding matrices, perturbed by an output perturbation mechanism $\M$ (\eg, GM);
the service provider runs $p(\cdot)$ on noisy (labeled) embeddings for fine-tuning (Section~\ref{sect:fine-tune}) or pre-training (Section~\ref{sect:pre-train}).
The challenge lies in analyzing $S_2(f)$ for different pipeline parts, which we address by normalizing $f(\cdot)$.

DP-Forward can be naturally used to protect \emph{inference} sequences (Section~\ref{sect:inference}), unlike DP-SGD.
It exploits the free post-processing (\ie, inference works on noisy embeddings), incurring minimal changes to pipelines with the extra ``plug-and-play'' noise layer.

\subsection{Notions of Sequence (Local) DP}
\label{sect:new_notion}
Embeddings $f(X)$ encode semantic information of input sequences $X$, each of which has $n$ tokens (Section~\ref{sect:pre_BERT}).
Fine-tuning (or subsequent inference of) NLP pipelines essentially processes $f(X)$.
DP-Forward fine-tuning protects \emph{every} $X$ by an output perturbation mechanism $\M$ over $f(X)$, in contrast to DP-SGD, which perturbs \emph{aggregates} of gradients $f'(X,y)$ over $X$ and label~$y$.
Simply put, our $(\epsilon,\delta)$-LDP holds for $X$ while DP-SGD provides CDP for $(X, y)$.

Sequence-only protection is meaningful since sequences often convey (implicit) sensitive information (\eg, authorship), whereas labels (\eg, a single bit denoting positive/negative) can be public.
We defer to Section~\ref{sect:cmp_dpsgd} for achieving ``full'' LDP over $(X, y)$.
To bridge the gap between theoretical guarantees of DP-SGD and~DP-Forward, we first define sequence DP\footnote{One could generalize it to ``feature'' (or ``input'') DP, as DP-Forward also allows other types of features beyond embeddings (and its essence is input-only privacy).
To keep our focus on NLP, we use ``sequence'' here.
(PixelDP~\cite{sp/LecuyerAG0J19} treats pixels as image features.)
} (SeqDP) in the \emph{central} setting.

\begin{definition}[SeqDP]
\label{def:token_seqldp}
\!\!For $\epsilon \geq 0, 0 \leq \delta \leq 1$, $\M$ is $(\epsilon,\delta)$-SeqDP, if $\forall \D \simeq \D'$ that only differ in a sequence at some index $i$: $(X_i, y_i) \in \D$ and $(X'_i, y_i) \in \D', \forall X_i,X'_i$, and any possible output subset $\O$,
\begin{align*}
    \Pr[\M(\D) \in \O] \leq e^\epsilon \Pr[\M(\D') \in \O] + \delta.
\end{align*}
\end{definition}

\subsubsection{Label DP}
The recently proposed notion of label DP~\cite{nips/GhaziGKMZ21, nips/EsmaeiliMPST21} is originally studied in PAC learning~\cite{colt/ChaudhuriH11}.
It only protects \emph{labels} (not the corresponding inputs/images): $(\epsilon, \delta)$-DP is only w.r.t. labels.

Our SeqDP is ``more secure'' than or at least ``complements'' label DP, which has an inherent flaw~\cite{nips/BusaSV21}:
As labels typically rely on their sequences (but not vice versa), it is very likely to recover the true labels from the raw sequences, even if the labels are protected (by any label-DP mechanism).
The follow-up~\cite{aistats/WuZWG23} shows the impossibility of label protection under label DP even with arbitrarily small $(\epsilon, \delta)$ when models generalize.
Moreover, labels can be absent (\eg, inference or self-supervised learning), for which SeqDP upgrades to the standard $(\epsilon,\delta)$-DP, whereas label DP is simply inapplicable.

\subsubsection{Sequence Local DP (SeqLDP)} 
We further define SeqLDP, the \emph{local} counterpart of sequence DP.
Note that the above discussion of label DP in relation to SeqDP also carries over to SeqLDP.

\begin{definition}[SeqLDP]
\label{def:seqldp}
\!\!For $\epsilon \geq 0, 0 \leq \delta \leq~1$, $\M$ satisfies~$(\epsilon,\delta)$-SeqLDP, if $\forall X, X'$ with the same $y$, and any possible output subset~$\O$,
\begin{align*}
    \Pr[\M(X,y) \in \O] \leq e^\epsilon \Pr[\M(X',y) \in \O] + \delta.
\end{align*}
\end{definition}

In theory, SeqLDP remains a strong notion (like the standard LDP).
It is meant to be information-theoretic protection on~sequence and bounds the indistinguishability of \emph{any} $X, X'$ (differing by up to $n$ tokens), and hence governing the ``usefulness'' of noisy embeddings.

\subsubsection{Sequence-Level SeqLDP vs. Token-Level SeqLDP}\label{sect:weakening}
In practice, as a strong notion balancing seemingly conflicting requirements (ideal theoretical guarantees and empirical utility), attaining a meaningful range of $\epsilon$ for SeqLDP is a struggle.
Adding Gaussian noise to the outputs of $f(\cdot)$ for $(\epsilon, \delta)$-SeqLDP requires bounding the $L_2$-sensitivity $S_2(f), \forall X, X'$.
Our approach is to \emph{normalize} the outputs (with extra benefits elaborated in Section~\ref{sect:splitting}), similar to clipping~gradients in DP-SGD.
It generally works better when $f(\cdot)$ has more layers 
(at the same meaningful range of~$\epsilon$)
since fewer (trainable) parameters/layers of $p(\cdot)$ are ``affected'' by the noisy outputs.

Unfortunately, when $f(\cdot)$ includes the first few layer(s), \eg, only the input embedding layer is available to the users (say, for saving user-side storage and computation overheads),
it leads to poor utility.
As a comprehensive study, we resort to \emph{row-wise} normalization with the (composition of) Lipschitz constants~\cite{icml/KimPM21} to maintain utility for those cases.\footnote{One
might also resort to the weaker random DP~\cite{jpc/HallWR13} -- $(\epsilon, \delta)$-DP holds on all but a small $\gamma$-proportion of ``unlikely'' $\D \simeq \D'$ for an extra parameter $\gamma \in (0,1)$.
It is useful when the global sensitivity is hard to compute.
Exploring it is left as future work.}
In contrast to the general normalization, it aims
for weaker SeqLDP at the \emph{token} level (\cf event-level vs. user-level LDP~\cite{sp/ZhouWCFS22}), a finer granularity in the ``protection hierarchy,'' protecting any \emph{neighboring} sequences (vs. datasets) differing in any single token (vs. sequence).
Details are deferred to
\ifnum\fullFlag=1
Appendix~\ref{apdx:token-level}.
\else
our full version.
\fi

\subsection{Our Approach for Sequence LDP}
\label{sect:splitting}

\ifnum\fullFlag=1
DP-Forward in our paper (except Appendix~\ref{apdx:token-level}) 
applies the general normalization approach 
to any $f(\cdot)$ for 
\emph{sequence-level} (Seq)LDP.
\fi

Let $f(\cdot)$ be an arbitrarily deep forward pass, 
ranging from the first (input embedding) layer itself to all but the last (task) layer in a BERT-based pipeline (Figure~\ref{fig:encoder}).
Correspondingly, let $p(\cdot)$ be the remaining layers, ranging from the last task layers themselves to all but the first (input embedding) layer.
Every sequence $X$ becomes an embedding matrix $f(X) \in \mathbb{R}^{n \times d}$ at the output of layers in encoders or $\mathbb{R}^{1 \times d}$ before task layers (Section~\ref{sect:pre_BERT}).
To offer $(\epsilon, \delta)$-SeqLDP, we adopt a suitable output perturbation mechanism $\M$, such as GM, considering that a dataset has only one labeled sequence.

Since $\M$ can work on the output of \emph{any} hidden layer, estimating $S_2(f)$ is non-trivial.
Specifically, $\MHA$ itself, let alone more layers included, is not Lipschitz continuous, meaning its outputs can change arbitrarily for even slight input variation~\cite{icml/KimPM21}.
To address this, our approach is to normalize or clip the function outputs:
$$||f(\cdot)||_F = C \ \text{or} \ f(\cdot) / \max(1, \frac{||f(\cdot)||_F}{C})$$
as in DP-SGD~\cite{ccs/AbadiCGMMT016}, where $C$ is a tunable parameter.
We then have $S_2(f) = 2C$.
Such normalization makes task utility less ``sensitive'' to the choice of $C$ since signal and noise increase proportionally with~$C$, whereas the signal may be unchanged when $f(\cdot)$ is not~clipped.
It also has many other benefits, such as stabilizing training, avoiding overfitting, and accelerating convergence~\cite{iclr/AboagyeZYWZWYP22}.
Hence, we resort to normalization in our experiments.
One can then calibrate Gaussian noise $Z$ and derive $f(X)+Z$ for the post-noise layers $p(\cdot)$.

Note that we remove the residual connection 
when adding noise to the output of the \emph{first} MHA layer to avoid $p(\cdot)$ reaccessing $X$ (dashed arrow, Figure~\ref{fig:encoder}) to maintain free post-processing.
This may lead to instability (\eg, gradient vanishing)~\cite{icml/YuZC0L21}, but it can be mitigated by pre-training new BERT without such a residual connection to keep consistent with later fine-tuning/inference.

DP-Forward using GM suffers from the ``curse of dimensionality'' when $d$ is large (\eg, $768$ for BERT-Base).
To alleviate the curse, we can append two linear maps, $M_1, M_2 \in \mathbb{R}^{d \times d'}$ with $d' \ll d$, such that $f(\cdot)$ and $p(\cdot)$ respectively have $M_1$ and $M_2$.
Both maps are randomly initialized and updated like other weights using gradients.
The raw embedding matrix is first right multiplied by $M_1$, leading to $\mathbb{R}^{n \times d'}$ or $\mathbb{R}^{1 \times d'}$, before being normalized.
Our privacy guarantee will not be affected since $S_2(f)$ remains the same.
We then use $M_2$ to restore the dimensionality to be compatible with the raw pipeline; $M_2$ incurs~no extra privacy loss due to the free post-processing.
Nevertheless, it needs dedicated efforts to modify the pipeline; dimension-reduced embedding matrices may also lose useful information, degrading task utility.
We thus make $M_1$ and $M_2$ optional (see Section~\ref{exp:configuration}).

\subsection{DP-Forward Fine-tuning}
\label{sect:fine-tune}

Suppose we use a raw, public BERT checkpoint\footnote{Using noisy BERT for fine-tuning (and subsequent inference) is deferred to Section~\ref{sect:pre-train}.} for fine-tuning. 
In the forward pass of the $i$-th ($i \geq 1$) step, it offers the latest $f^{(i-1)}(\cdot)$ to a \emph{batch} of users, mimicking the regular mini-batch SGD.
$f^{(0)}$ is from the raw checkpoint.
Users are randomly chosen (without replacement), and their number is a fixed parameter.
Users in the batch individually compute their noisy embeddings $f^{(i-1)}(X)+Z$ to ensure SeqLDP (Theorem~\ref{theo:seq-ldp}).
They then send them with unperturbed labels $y$ to the service provider, who runs $p^{(i-1)}(\cdot)$ over $(f^{(i-1)}(X)+Z, y)$ to compute the batch loss; any post-processing of embeddings under SeqLDP incurs no extra privacy degradation on $X$.
$p^{(0)}$ here includes the rest raw BERT part and randomly initialized task layers. 

During the back-propagation, the service provider can update $p^{(i-1)}(\cdot)$ to $p^{(i)}(\cdot)$ via the gradient (derived from the loss and noisy embeddings) of the post-noise layers.
To avoid accessing users' raw $X$, it needs to freeze the pre-noise layers $f^{(i-1)}(\cdot)$ as $f^{(0)}$.
Parameter freezing is compatible with the more recent zero-shot or in-context learning paradigm~\cite{emnlp/MinLHALHZ22}.
It is useful when models are gigantic and full fine-tuning is expensive.
However, the more layers are frozen, the worse the utility might be (even in non-private settings).

There are two general ways to update $f^{(i-1)}(\cdot)$ securely:
i) We can assume an extra trusted party (as in DP-SGD), but it becomes central DP.
ii) Users can first derive the gradients for the layers inside $f^{(i-1)}(\cdot)$ locally on their $X$ and then resort to secure aggregation~\cite{ccs/BonawitzIKMMPRS17} for global updates at the service provider.
However, it is costly.
For better utility, we update $f^{(i-1)}(\cdot)$ in experiments, requiring us to consider privacy degradation across different \emph{epochs} due to the composability (as detailed below).
Dedicated approaches (that balance efficiency, privacy, and utility) are left as future work.

\begin{theorem}
\label{theo:seq-ldp}
Let $f(\cdot)$ be the pre-noise function (of BERT-based pipelines) and $\M$ be GM with $\epsilon \geq 0, 0 \leq \delta \leq 1$.
DP-Forward fine-tuning running $\M$ on normalized/clipped $f(\cdot)$ ensures $(\epsilon, \delta)$-SeqLDP.
\end{theorem}

The proof follows that of GM~\cite{fttcs/DworkR14}.
The crux is that $S_2(f), \forall X, X'$ is given by the output normalization, independent of the inputs.

\smallskip
\noindent \textbf{Privacy Accounting.}
An epoch refers to an entire transit of the private training corpus.
Every $X$ is used once per epoch.
The number of epochs $k$ is a hyperparameter, which is typically small.
Repeated applications of GM over the same $X$ ask for estimating the overall privacy loss due to the composability (unless freezing $f$ for re-using $f(X)+Z$).
The well-known moments accountant~\cite{ccs/AbadiCGMMT016}~(or~its generalization to R\'{e}nyi DP~\cite{csfw/Mironov17}) only provides a~loose \emph{upper} bound, which is even inapplicable if unbounded moments exist.
Gaussian DP~\cite{corr/BuDLS19} proposes an accountant based on the central limit theorem.
Yet, it leads to significant underestimation by a \emph{lower} bound.~Instead, we resort to a recent numerical accountant~\cite{nips/GopiLW21}, which outperforms RDP or GDP by approximating the true overall $\epsilon$ to arbitrary accuracy.
It composes the privacy curve of a mechanism by truncating and discretizing PLRVs with their PDFs convoluted by FFT~\cite{nips/GopiLW21}.

\subsection{DP-Forward with Shuffling versus DP-SGD}
\label{sect:cmp_dpsgd}

DP-Forward ensures SeqLDP for fine-tuning, while DP-SGD offers \emph{central} DP (for sequence-label pairs).
To facilitate a fair comparison (on privacy-utility tradeoffs), we make two changes.
First, we also perturb the labels with a suitable mechanism for the standard LDP, \ie, extending the protection from sequence to sequence-label pairs.
Second, we use shuffling~\cite{soda/ErlingssonFMRTT19} to ``translate'' our (label-protected) DP-Forward with LDP to claim (example-level) CDP as DP-SGD.

\smallskip
\noindent \textbf{Discrete Labels Perturbation.}
For most NLP tasks, \eg, bi-/multi-nary classification in the GLUE benchmark~\cite{iclr/WangSMHLB19}, the size $|\mathbf{y}|$ of label space is often small.
A simple yet effective solution for discrete data is randomized response (RR)~\cite{jasa/Warner65} proposed decades ago!
Specifically, RR perturbs a true label $y$ to itself $\hat{y} = y$ with the probability 
$$\Pr[y = \hat y] = e^\epsilon/(e^\epsilon + |\mathbf{y}| -1),$$
or to $\forall \hat y \in \mathbf{y} \setminus y$ uniformly, where $\mathbf{y}$ denotes the label space.

When $|\mathbf{y}|$ is large, we can use prior to ``prune'' $\mathbf{y}$ to smaller $\mathbf{y}'$~\cite{nips/GhaziGKMZ21}.
The prior can be publicly available (\eg, auxiliary corpora similar to the users' data) or progressively refined from a uniform distribution via the multi-stage training~\cite{nips/GhaziGKMZ21}.
One can then estimate an optimal $|\mathbf{y}'|$ by maximizing the probability that the output is correct, \ie, $\Pr[y = \hat y]$.
With (prior-aided) RR~\cite{nips/GhaziGKMZ21}, we can achieve full LDP.

\begin{theorem}
\label{theo:LDP_fine-tuning}
Let $f(\cdot)$ be the pre-noise function (of BERT-based pipelines), $\M$ be GM with $\epsilon_1 \geq 0, 0 \leq \delta \leq 1$, and $\M_{RR}$ be (prior-aided) RR with $\epsilon_2 \geq 0$.
DP-Forward fine-tuning perturbing $f(X)$ and $y$ separately by $\M$ and $\M_{RR}$ ensures $(\epsilon_1+\epsilon_2, \delta)$-LDP.
\end{theorem}

The proof follows from the basic composition theorem~\cite{fttcs/DworkR14}.

\smallskip
\noindent \textbf{Privacy Amplification by Shuffling.}
If noisy embedding-label pairs are also shuffled properly, DP-Forward can claim example-level CDP (as in DP-SGD), which ``amplifies'' LDP guarantees by $\Theta(\sqrt{N})$ for a total number of $N$ users (without extra noise addition)~\cite{soda/ErlingssonFMRTT19}.
We then show that DP-Forward qualitatively outperforms DP-SGD from the SNR perspective under a similar privacy regime.

Suppose we train for an epoch, and the normalization~factor is $C$.
For DP-SGD, the batch size is $b$; the subsampling probability and the number of training steps are respectively $b/N$ and $N/b$.
If each step is $(\epsilon,\delta)$-DP, the overall privacy loss is $(O(\epsilon \sqrt{b/N}), \delta)$-DP using the strong composition and privacy amplification by subsampling~\cite{ccs/AbadiCGMMT016}.

DP-Forward with shuffling can also be seen as composing~$N$ subsamplings, each a fraction of size $1$~\cite{corr/Thomas22}.
It is $(O(\epsilon \sqrt{1/N}), \delta)$-DP, which is ``amplified'' from $(\epsilon,\delta)$-LDP.
For an easier analysis of SNR, we omit $\epsilon_2$ of RR since the overall $\epsilon$ is dominated by composing subsampled Gaussian.
So, our Gaussian noise variance is $b \times$ smaller than DP-SGD's in each step;
the SNR of each entry in embeddings vs. the aggregation of $b$ gradients can be estimated as $O(C / \sqrt{nd})$ for DP-Forward vs. $O(C / \sqrt{d'})$ for DP-SGD, where $d'$ is the gradient dimension and is much larger than $nd$, the embedding-matrix size.

\subsection{DP-Forward Inference}
\label{sect:inference}

Given only fine-tuned pipeline parts $f(\cdot)$, users can derive the noisy embedding matrices of their \emph{test} sequences for inferences at the service provider while ensuring $(\epsilon, \delta)$-LDP.
Inference using noise aligned to the noisy fine-tuning is also beneficial for task accuracy.

Local inference (as in DP-SGD) without noise forces the service provider to reveal its \emph{entire} pipeline, losing its intellectual property and incurring more time and storage costs for both $f(\cdot)$ and $p(\cdot)$.

\begin{theorem}
\label{theo:inference_ldp}
Let $f(\cdot)$ be the fine-tuned pre-noise layers (of BERT-based pipelines) and $\M$ be GM with $\epsilon \geq 0, 0 \leq \delta \leq 1$.
DP-Forward inference running $\M$ on normalized/clipped $f(\cdot)$ ensures $(\epsilon, \delta)$-LDP.
\end{theorem}

The proof is inherited from GM~\cite{fttcs/DworkR14}.
Different from DP-Forward fine-tuning, LDP holds for test sequences since the labels are absent.

\subsection{DP-Forward Pre-training} 
\label{sect:pre-train}
Directly using the raw BERT might not ``match'' DP-Forward fine-tuning/inference, degrading task utility.
Pre-training BERT with DP-Forward on publicly available text (\eg, Wikipedia), besides the private user-shared data, can make future operations ``adaptive'' to noise.
It requires us to modify the raw MLM objective in Eq.~(\ref{eq:mlm}):
$$L_{\MLM}^* = - \sum_{i \in \mathcal{I}} \log \Pr[x_i|\M(f(\hat{X}));\theta^*],$$
where $\theta^*$ denotes the parameters of ``noisy'' BERT.
This endows the noisy BERT with some ``de-noising'' ability since the objective is to predict the raw masked tokens from noisy embeddings $\M(f(\hat{X}))$.~It does not really breach privacy due to the free post-processing; LDP is ensured for each sequence, as the pre-training is self-supervised (without labels).
Such noisy pre-training can also be outsourced to dedicated GPU clusters, enabling ``de-noising BERT as a service.''

De-noising as post-processing is not new, but most prior arts need prior knowledge, \eg, Bayesian prior.
aGM formulates it as an unusual estimation problem since a single noisy output is observed for each input, which can then be solved by appropriate estimators, \eg, the Bayesian one~\cite{icml/BalleW18}.
Another attempt~\cite{sp/LecuyerAG0J19} trains a separate noisy auto-encoder, which learns the identity function $f(X) = X$ stacked before an image classification network, to de-noise the noisy input.
It has limited applications for only noisy input embeddings and incurs extra changes when migrating it to an NLP pipeline.

\section{Optimizing Matrix Gaussian Noise}
\label{sect:amgm}

To instantiate $\M$ for $f(\cdot) \in \mathbb{R}^{n \times d}$ of DP-Forward, a natural question is whether the classical GM is optimal.
The answer is no.
Its privacy analysis applies a \emph{sufficient but not necessary} condition for $(\epsilon, \delta)$-DP while using Gaussian tail approximations, and its variance formula cannot extend to $\epsilon > 1$ for a single run (\eg, inference)~\cite{fttcs/DworkR14}.

Another candidate is the matrix-variate Gaussian (MVG) mechanism~\cite{ccs/ChanyaswadDPM18}, tailored for matrix-valued functions.
It exploits possibly \emph{non-i.i.d.} noise from a \emph{matrix} Gaussian distribution and outperforms GM in several usage cases~\cite{ccs/ChanyaswadDPM18}.
Yet, it is not optimal either,~with the root cause still being based on a sufficient DP condition (Section~\ref{sect:mvg}).
To improve it, we resort to a necessary and sufficient condition from aGM~\cite{icml/BalleW18} for calibrating the matrix Gaussian noise
(Section~\ref{sect:agm-aMGM}).

\subsection{Matrix-Variate Gaussian (MVG) Mechanism}
\label{sect:mvg}

In contrast to the classical GM, MVG adopts possibly \emph{non-i.i.d.} noise $Z \in \mathbb{R}^{n \times d}$ drawn from the zero-mean matrix Gaussian distribution $\MN_{n,d}(0,\Sigma,\Psi)$, where $\Sigma \in \mathbb{R}^{n \times n}$ and $\Psi \in \mathbb{R}^{d \times d}$ are the row- and column-wise covariance matrices.
Intuitively, it adds less noise to more ``important'' rows or columns for possible better utility.

\begin{definition}[Matrix Gaussian Distribution]
\label{def:MGD}
The PDF for an $n \times d$ random variable $Z$ following $\MN_{n,d}(0,\Sigma,\Psi)$ has the form:
\begin{align}
    \Pr{(Z|0,\Sigma,\Psi)} = \frac{\exp\left(-\frac{1}{2}||U^{-1}ZV^{-\top}||^2_F\right)}{(2\pi)^{nd/2}|\Psi|^{d/2}|\Sigma|^{n/2}},
\end{align}
where $U \in \mathbb{R}^{n \times n}$ and $V \in \mathbb{R}^{d \times d}$ are invertible with $\Sigma = UU^\top$ and $\Psi = V V^\top$, and 
$|\cdot|$ denotes the matrix determinant~\cite{cu/HJ2012}.
\end{definition}

The definition is equivalent to the conventional form given by the matrix trace.
It generalizes the univariate Gaussian used in GM; $Z$ becomes i.i.d. when $\Sigma, \Psi$ are diagonal and equal-valued.
Below recites the main theorem of the MVG mechanism for $(\epsilon,\delta)$-DP.

\begin{theorem}[The MVG Mechanism with $(\epsilon,\delta)$-DP~\cite{ccs/ChanyaswadDPM18}]
\label{theo:MVG}
Let 
\begin{align*}
    \sigma(\Sigma^{-1}) & = [\sigma_1(\Sigma^{-1}), \ldots, \sigma_n(\Sigma^{-1})]^\top , \\ 
    \sigma(\Psi^{-1}) & = [\sigma_1(\Psi^{-1}), \ldots, \sigma_d(\Psi^{-1})]^\top
\end{align*}
be the vectors of (non-increasingly ordered) singular values of $\Sigma^{-1}$ and $\Psi^{-1}$, respectively.
The MVG mechanism using noise from the~matrix Gaussian distribution $\MN_{n,d}(0, \Sigma, \Psi)$ satisfies $(\epsilon,\delta)$-DP if 
\begin{align*}
    ||\sigma(\Sigma^{-1})||_2 \cdot ||\sigma(\Psi^{-1})||_2 \leq \frac{\left(-\beta +\sqrt{\beta^2 +8 \alpha \epsilon}\right)^2}{4\alpha^2},
\end{align*}
where $\alpha = [H_r +H_{r,1/2}]\gamma^2+2H_r\gamma S_2(f)$, $\beta = 2(nd)^{1/4}H_rS_2(f)\zeta(\delta)$, with $H_r$ (or $H_{r,1/2}$) being the generalized harmonic number of order $r$ (of $1/2$), $\gamma$ being $\sup_{\D}||f(\D)||_F$, and $\zeta(\delta)=2\sqrt{-nd\ln{\delta}}-2\ln{\delta}+nd$.
\end{theorem}

To illustrate how the MVG mechanism works, 
we quote an example~\cite{ccs/ChanyaswadDPM18}: performing regression using an identity query on a liver disorders dataset~\cite{prl/McDermottF16} with $6$ features and $248$ samples (\ie, $f \in \mathbb{R}^{248 \times 6}$).
MVG treats `ALT' and a teacher label as the two most indicative~features based on some prior, thus added with less noise~\cite{ccs/ChanyaswadDPM18}.
To report the best \emph{empirical} results, it tries different precision budget (or noise variance) allocation strategies so that the total budget (Theorem~\ref{theo:MVG}) is not overspent.
For example, it evenly allocates $\tau > 50 \%$ (a tunable parameter) of the budget to the two important features and the rest to the other four.
Compared to GM using i.i.d. Gaussian noise, MVG can improve root mean square error (RMSE) by up to $0.003$ at the same privacy level~\cite{ccs/ChanyaswadDPM18}.

\smallskip
\noindent \textbf{Sub-optimality of MVG.}
Theorem~\ref{theo:MVG} presents an upper bound~on the product of~$L_2$-norms of two singular-value vectors $\sigma(\Sigma^{-1})$ and $\sigma(\Psi^{-1})$, assuming $||f(\D)||_F$ is bounded for any $\D$ by a constant~$\gamma$.
The upper bound monotonically decreases with $\beta$ that depends on $nd$ and approaches $0$ as $nd \rightarrow \infty$, making the sums of noise variances large.
A similar situation exists in high privacy regimes~$\epsilon \rightarrow 0$.

At least two slacks caused the sub-optimality.
The first and foremost is due to a \emph{sufficient condition} for $(\epsilon,\delta)$-DP~\cite{fttcs/DworkR14}: $\Pr[\L_{\M,\D,\D'} \geq \epsilon] \leq \delta$, which is also used in the classical GM.
With the Laurent-Massart Theorem~\cite{as/LaurentM00}, MVG further transforms it to $\Pr[\L_{\M,\D,\D'} \leq \epsilon] = 1$ for a subset of all the possible outputs.
The second lies in a loose matrix-trace-based privacy analysis; a follow-up~\cite{tmc/YangXYWGLL23} derives a tighter bound from Definition~\ref{def:MGD} and a matrix-norm inequality.

\begin{algorithm}[t!]
\caption{$A(\epsilon, \delta)$: Derive the Upper Bound $\B$}
\label{alg:up_bound}
\LinesNumbered
\KwIn{Privacy parameters $\epsilon, \delta$}
\KwOut{$\B$}
Let $\delta_0 = \Phi(0) - e^\epsilon \Phi(-\sqrt{2\epsilon})$\label{line1}\;
\eIf{$\delta \geq \delta_0$}{
Re-write $g^+_\epsilon(v) = \Phi(\sqrt{\epsilon v})- e^\epsilon\Phi(-\sqrt{\epsilon (v+2)}$\;
Compute $v^* = \sup\{v \in \mathbb{R}_{\geq 0}: g^+_\epsilon(v) = \delta \}$\;
Let $\alpha = \sqrt{1+v^*/2}-\sqrt{v^*/2}$\;
}
{$g^-_\epsilon(u) = \Phi(-\sqrt{\epsilon u})- e^\epsilon\Phi(-\sqrt{\epsilon (u+2)})$\;
Compute $u^* = \inf \{u \in \mathbb{R}_{\geq 0}:g^-_\epsilon(u) = \delta \}$\;
Let $\alpha = \sqrt{1+u^*/2} + \sqrt{u^*/2}$\;
}
Return $\B = \sqrt{2\epsilon}/\alpha$
\end{algorithm}

\subsection{Analytic Matrix Gaussian Mechanism}
\label{sect:agm-aMGM}

To enhance MVG while still adding possibly non-i.i.d. noise $Z \sim \MN_{n,d}(0,\Sigma,\Psi)$, we put forth the analytic matrix Gaussian mechanism (aMGM) by exploiting a \emph{necessary and sufficient} condition~for $(\epsilon,\delta)$-DP, which is formulated using two PLRVs by the analytic~GM (aGM)~\cite{icml/BalleW18}.
It is non-trivial\footnote{A recent pre-print~\cite{corr/YangXLLW21} also studied using matrix Gaussian distribution.
The~proof~of \cite[Lemma~$4$]{corr/YangXLLW21}, pivotal for our Theorem~\ref{theo:dp_amgm}, is problematic.
\ifnum\fullFlag=1
We prove it in Appendix~\ref{apdx_sectB}.
\else
We prove it in our full~version.
\fi}
since we now need to work with two covariance matrices $\Sigma$ and $\Psi$ instead of a single variance $\sigma^2$ in aGM.

\begin{theorem}[\cite{icml/BalleW18}]
\label{suff_and_nece}
A mechanism $\M$ is $(\epsilon,\delta)$-DP iff, $\forall \D \simeq \D'$,
\begin{align}
\label{eq:suff_and_nece}
    \Pr[\L_{\M,\D,\D'} \geq \epsilon] - e^\epsilon \Pr[\L_{\M,\D',\D} \leq -\epsilon] \leq \delta.
\end{align}
\end{theorem}

It directly implies the sufficient condition due to $\Pr[\L_{\M,\D',\D} \leq -\epsilon] \geq 0$.
We next show that $\L_{\M,\D,\D'}$ or $\L_{\M,\D',\D}$ of aMGM is also Gaussian, a similar result has been proven in aGM~\cite[Lemma~3]{icml/BalleW18}.

\begin{lemma}
\label{lemma_4}
The PLRVs of our aMGM follow a distribution $\N(\eta,2\eta)$ with $\eta = \frac{||U^{-1}\Delta V^{-\top}||^2_F}{2}$, where $\Delta = f(\D) - f(\D')$.
\end{lemma}

With Lemma~\ref{lemma_4}, we can then specialize the left-hand side of Eq.~(\ref{eq:suff_and_nece}).
Particularly, we use the Gaussian cumulative density function~(CDF)
$$\Phi(t) = \Pr[\N(0,1) \leq t] = \frac{1}{\sqrt{2\pi}} \int^t_{-\infty} e^{-y^2/2} dy$$
to explicitly express the two probabilities (see Lemma~\ref{lemma_5}) instead of approximating them by the tail bounds of a Gaussian distribution.

\begin{lemma}
\label{lemma_5}
For any $\D \simeq \D'$, let $\Delta' = U^{-1}\Delta V^{-\top}$ with $\Delta = f(\D) - f(\D')$.
The following holds for any $\epsilon \geq 0$:
\begin{align*}
    \Pr[\L_{\M,\D,\D'} \geq \epsilon] & = \Phi\left(\frac{||\Delta'||_F}{2} - \frac{\epsilon}{||\Delta'||_F}\right), \\
    \Pr[\L_{\M,\D',\D} \leq -\epsilon] & = \Phi\left(-\frac{||\Delta'||_F}{2} - \frac{\epsilon}{||\Delta'||_F}\right).
\end{align*}
\end{lemma}

We can further re-write the left-hand side of Eq.~(\ref{eq:suff_and_nece}) as $g(||\Delta'||_F)$:
\begin{align}
\label{privacy_curve} \Phi\left(\frac{||\Delta'||_F}{2}-\frac{\epsilon}{||\Delta'||_F}\right) - e^\epsilon \Phi\left(-\frac{||\Delta'||_F}{2}-\frac{\epsilon}{||\Delta'||_F}\right),
\end{align}
a function of $\Delta$ and $(\Sigma , \Psi)$; it is defined w.r.t. $\Delta$ and $\sigma^2$ for aGM~\cite{icml/BalleW18}.~To satisfy Theorem~\ref{suff_and_nece}, we require 
$g(||\Delta'||_F) \leq \delta, \forall \D \simeq \D'$.
Since $g(\cdot)$ is \emph{monotonically increasing}~\cite[Lemma~7]{icml/BalleW18}, we first find the upper bound $\B$ of $||\Delta'||_F$ as the ``solution'' to $g(||\Delta'||_F) = \delta$ and then determine $U,V$ (hence $\Sigma,\Psi$) based on $\B$ and $\Delta$ with $S_2(f) = \sup_{\D \simeq \D'}||\Delta||_F$.

\subsubsection{Computing the upper bound $\B$}
One could derive an analytic expression for $\B$ using the tail bounds of $\Phi(t)$, which is sub-optimal due to the slack in the tail bounds.
Instead, we adapt a ``numerical solver,'' as detailed in Alg.~\ref{alg:up_bound}, for $\B$ since $\Phi(t)$ can also be represented by $(1+\erf(t/\sqrt{2}))/2$, where $\erf$ is the standard error function.\footnote{Its efficient implementation to extremely high accuracy is supported in most statistical and numerical software packages, \eg, Python math library.}

For the first term of Eq.~(\ref{privacy_curve}), its input $||\Delta'||_F/2-\epsilon/||\Delta'||_F$ changes sign at $||\Delta'||_F = \sqrt{2\epsilon}$, while the other term's input $-||\Delta'||_F/2-\epsilon/||\Delta'||_F$ is always negative.
Therefore, we only 
consider $||\Delta'||_F = \sqrt{2\epsilon}/\alpha$ under two cases $0 < \alpha \leq 1$ and $\alpha > 1$ for a variable $\alpha$.

When $\alpha = 1$, $\delta_0 = g(\sqrt{2\epsilon})$ in line~$1$.
If $\delta \geq \delta_0$ (or $0 < \alpha \leq 1$), we can use $v = (1/\alpha - \alpha)^2/2$ to re-write $g(\cdot)$ as $g^+_\epsilon(v)$ (line~3).
For $\alpha>1$, we can use $u = (\alpha - 1/\alpha)^2/2$ to re-write $g(\cdot)$ as $g^-_\epsilon(u)$ (line~7).
In either case, given the ``oracle'' computing $\Phi(t)$ via $\erf$, we derive $u^*$ or $v^*$ using Newton's method, recover $\alpha$, and return $\B = \sqrt{2\epsilon}/\alpha$.

\subsubsection{Determining the covariance matrices $\Sigma = UU^\top$ and $\Psi = VV^\top$}

\ifnum\fullFlag=1
With Lemma~\ref{lemma_3},
\else
With~\cite[Lemma~$4$]{corr/YangXLLW21},
\fi
and let $\sigma_i(\cdot)$ be the $i^{\text{th}}$ singular value; we have 
\begin{align}
\label{eq:from_lemma3}
||U^{-1} \Delta V^{-\top}||^2_F \leq \sum^{\min \{n, d\}}_{i=1}\sigma^2_i(U^{-1})\sigma^2_i(\Delta)\sigma^2_i(V^{-\top}).
\end{align}
Since
$\sigma_i(U^{-1}) = 1/\sigma_{n-i+1}(U)$ and $\sigma_i(V^{-\top}) = 1 / \sigma_{d-i+1}(V)$ with $i \in [1,r]$,
we transform the right-hand side of Eq.~(\ref{eq:from_lemma3}) to
\begin{align*}
\sum^r_{i=1}\frac{\sigma^2_i(\Delta)}{\sigma^2_{n-i+1}(U)\sigma^2_{d-i+1}(V)} 
\leq \frac{\sum^r_{i=1} \sigma^2_i(\Delta)}{\sigma^2_{n}(U)\sigma^2_{d}(V)}
= \frac{||\Delta||^2_F}{\sigma^2_{n}(U)\sigma^2_{d}(V)},
\end{align*}
\ifnum\fullFlag=1
where the inequality follows from Theorem~\ref{theo:SVD}, \ie, $\sigma_1(\cdot) \geq \cdots \geq \sigma_r(\cdot)$, and the last equality is directly from Lemma~\ref{lemma_2}.
\else
where the inequality follows from $\sigma_1(\cdot) \geq \cdots \geq \sigma_r(\cdot)$ and the last equality is directly from Lemma~$5$ in our full version.
\fi

Given $\B \geq ||U^{-1} \Delta V^{-\top}||_F $, it suffices to let $||\Delta||_F / \sigma_{n}(U)\sigma_{d}(V) \leq \B$ with $ \Delta = f(\D)- f(\D'),\forall \D \simeq \D'$.
Recall that $S_2(f)$ is the upper bound on $||\Delta||_F, \forall \D \simeq \D'$, we now reach the main theorem.

\begin{theorem}
\label{theo:dp_amgm}
Our aMGM satisfies $(\epsilon,\delta)$-DP, iff
$$\frac{S_2(f)}{\B} \leq \sigma_n(U) \sigma_d(V),$$
where $\B=A(\epsilon, \delta)$ as in Alg.~\ref{alg:up_bound}, $S_2(f)$ is the $L_2$-sensitivity, $\sigma_n(U)$ and $\sigma_d(V)$ are respectively the smallest singular values of $U$ and $V$.
\end{theorem}

Theorem~\ref{theo:dp_amgm} only constrains the lower bound on the product of $\sigma_n(U)$ and $\sigma_d(V)$, the two \emph{smallest} singular values; it offers infinite choices for all the others with the design space for $(\Sigma, \Psi)$ even larger than that of MVG (Theorem~\ref{theo:MVG}).
More importantly, the lower bound is independent of $nd$, which can lead to orders-of-magnitude variance reduction than MVG, confirmed by our experiments in Section~\ref{sect:exp}.
For $\epsilon \rightarrow 0$, we can still derive a valid $\B$ from $2\Phi^{-1}((1 + \delta)/2)$.

To determine $\Sigma, \Psi$, another implicit constraint is to keep smaller noise for better utility.
Let us first consider $\Sigma = U U^\top$.
Since it~is positive definite, we can also decompose it into $W_\Sigma \Lambda_\Sigma W^\top_\Sigma$; we then have $U = W_\Sigma \Lambda^{1/2}_\Sigma$, where $\Lambda^{1/2}_\Sigma = \{\sigma_i(U)\}^n_{i=1}$ specifies the row-wise noise magnitudes.
Assuming that the smallest overall noise will yield the best utility, we let all the singular values be the smallest: $\sigma_1(U) = \cdots = \sigma_n(U)$.
As $W_\Sigma$ can be any unitary matrix, we simply use the standard basis, resulting in $U = \sigma_n(U) \cdot I_n$ for an $n \times n$ identity matrix $I_n$ and hence the final $\Sigma$.
Similarly, we can pick $\Psi = V V^\top$ with $V = \sigma_d(V) \cdot I_d $, where $I_d$ is a $d \times d$ identity matrix.

\subsubsection{Drawing the noise $Z$}
With $\Sigma$ and $\Psi$, the last step is to draw~$Z$.
Pragmatically, we adopt the affine transformation below.

\begin{lemma}[\cite{ccs/ChanyaswadDPM18}]
\label{lemma_6}
Let $Z' \in \mathbb{R}^{n \times d}$ be a random variable with each entry i.i.d. drawn from the standard normal distribution.
The transformed $Z = U Z' V^\top$ follows $\MN_{n,d}(0,UU^\top,VV^\top)$.
\end{lemma}

Hence, we can first sample $nd$ i.i.d.~values from $\N(0,1)$ to form~$Z'$, then employ the transformation $UZ'V^\top$ such~that
\begin{align}
\label{eq:draw_Z}
    Z \sim \frac{S_2(f)}{\B} \MN_{n,d}(0, I_n, I_d).
\end{align}

\paragraph{Discussion.}
When instantiating DP-Forward using aMGM, we set $\sigma_1(U) = \cdots = \sigma_n(U)$ and $\sigma_1(V) = \cdots = \sigma_d(V)$ such that the row- and column-wise noises are the smallest, and our pilot experiments show this yields optimal task utility; aMGM actually ``degenerates'' to aGM with i.i.d. noise.
Nevertheless, aMGM also allows non-i.i.d. noise like MVG: 
By tuning the corresponding singular values larger, we can add more noise to the rows/columns that negatively impact the utility.
It might be helpful when $f(\cdot)$ (\eg, linear regression~on a small liver dataset~\cite{ccs/ChanyaswadDPM18}) is simple or $p(\cdot)$ does not ``mix up'' noisy rows/columns.
In contrast to our empirical approach (like MVG), one could theoretically formulate the allocation of singular values as optimization problems that maximize different utility functions tailored to applications.
It might outperform our uniform treatment but takes more dedicated efforts, which we leave as future work.

\section{Experiments}
\label{sect:exp}
\subsection{Experimental Setup}
\label{exp:setup}

\begin{figure*}[!t]
	\centering
	\includegraphics[width=\linewidth]{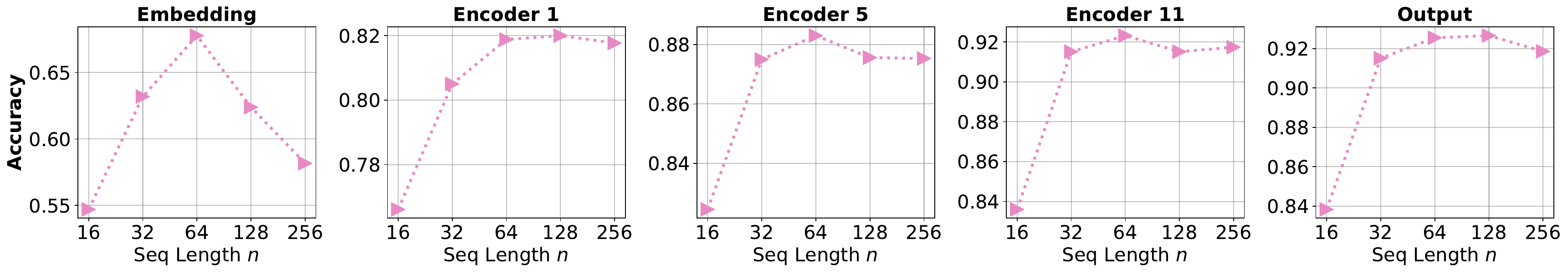}
	\vspace{-20pt}
	\caption{\mbox{SST-2} accuracy when tuning $n$ with noise added at different positions ($\epsilon = 16$ for SeqLDP)}
	\vspace{-8pt}
	\label{fig:embedding_paratuning}
\end{figure*}

\begin{figure*}[!t]
	\centering
	\includegraphics[width=\linewidth]{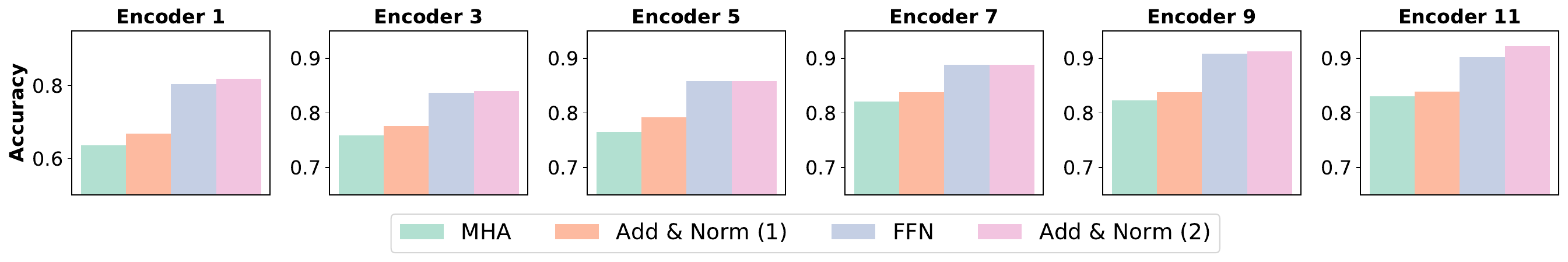}
	\vspace{-20pt}
	\caption{\mbox{SST-2} accuracy with noise added to the outputs of different \emph{sub-layers} in six encoders ($\epsilon = 16$ for SeqLDP)}
	\label{fig:inside_encoders}
\end{figure*}

We use three typical datasets/tasks that are widely used in NLP/DP literature~\cite{icml/YuZC0L21,iclr/LiTLH22,iclr/YuNBGIKKL22,acl/YueDWLSC21} and GLUE benchmark~\cite{iclr/WangSMHLB19}: 
i) Stanford sentiment treebank (\mbox{SST-2}), 
ii) Internet movie database (IMDb)~\cite{acl/MaasDPHNP11} for binary sentiment classification of single- and multi-sentence movie reviews, 
and
iii) Quora question pairs (QQP) for semantic equivalence test over question pairs on Quora.com.
Their test sets do not have any labels;
we use the original dev sets as the test sets.
Table~\ref{tbl:task_stat} summarizes their characteristics.
They all carry privacy risks; \eg, stylistic features of posts may leak the author's identity.
We use task \emph{accuracy} (w.r.t. the ground truth labels) as the utility~metric.

\begin{table}[!t]
\centering
\caption{Statistics of the downstream task datasets}
\begin{tabular}{l||r|r|r}
\toprule
 & \mbox{SST-2}~\cite{iclr/WangSMHLB19} & IMDb~\cite{acl/MaasDPHNP11} & QQP~\cite{iclr/WangSMHLB19} \\ 
\midrule 
\#train samples & $67,349$ & $25,000$ & $363,846$ \\ 
\#test samples & $872$ & $25,000$ & $40,320$ \\ 
\bottomrule 
\end{tabular}
\label{tbl:task_stat}
\end{table}

\smallskip
\noindent\textbf{Baselines.}
We instantiate $\M$ in DP-Forward by the classical GM, MVG~\cite{ccs/ChanyaswadDPM18}, and aMGM.
If not specified, all the results are based on aMGM.
For MVG, we adopt its \emph{unimodal} type, applicable to asymmetric functions like pre-noise layers $f(\cdot)$.
Specifically, we make the row-wise noise directional and assign the same precision budget to each row, assuming that tokens share the same importance.

By default, we report the accuracy of DP-Forward inferences on tasks fine-tuned using DP-Forward (with ${\sim}2$pp gains compared to the case of ``DP-Forward fine-tuning + non-private inference'').
We also realize DP-SGD fine-tuning with the latest Opacus~\cite{opacus} but do not add any noise to its inference.
Another baseline is non-private (in both fine-tuning and inference).

\smallskip
\noindent\textbf{Implementation.}
We run experiments on a cluster with Tesla~P$100$ GPUs.
We implement all the mechanisms and baselines in Python.
We use a raw BERT checkpoint \texttt{bert-base-uncased}~\cite{bert-base-uncased}, available in the Huggingface transformers library, for fine-tuning (Section~\ref{sect:fine-tune}) or further pre-train it over WikiCorpus (Section~\ref{sect:pre-train}).

For the hyperparameters used throughout our experiments, we set the number of training epochs $k = 3$, learning rate $\eta = 2\mathrm{e}{-5}$, batch size $b = 32$, and normalization/clipping factor $C = 1$.
We keep others (\eg, no weight decay, no learning rate decay) default as literature~\cite{iclr/WangSMHLB19}.
The privacy parameter $\delta$ is fixed as $1\mathrm{e}{-5}$~\cite{iclr/YuNBGIKKL22}.

\subsection{Configuring Matrix Dimensions}
\label{exp:configuration}

The sequence length $n$ is variable.
While the hidden dimensionality $d$ is tied as $768$ for BERT-Base, we can resort to two linear maps for ``mediating'' it (see Section~\ref{sect:splitting}).
Since we normalize embedding matrices of size $n \times d$ to have a fixed norm $C$, each entry's signal magnitude relies on $(n, d)$.
In contrast, the noise variance is the same given~$C$ and fixed privacy parameters.
The signal-to-noise ratios (SNRs) affecting accuracy can be configured based on $(n, d)$.

Figure~\ref{fig:embedding_paratuning} shows the evaluation accuracy of \mbox{SST-2} fine-tuned using DP-Forward with $n$ tuning from $16$ to $256$.
We study adding aMGM noise at five hidden layers' outputs.
The results indicate that the best accuracy is often achieved at $n = 64$ or $128$, so we opt for $n = 128$ (which is sufficient for most sequences) in subsequent~experiments.

We fine-tuned \mbox{SST-2} on noisy output embeddings under different choices of $\epsilon$ and reduced $d$.
Table~\ref{tbl:result_n} summarizes the results.
Reducing $d$ leads to larger SNRs (under fixed $C$ and $n$) but may also lose useful information, degrading accuracy.
For the same $\epsilon$, most accuracy variations are within $2$pp under different 
choices of $d$.
Balancing everything, we use the raw $d = 768$ in later experiments such that \emph{no} extra changes (including two linear maps) are made to pipelines.

\begin{table}[!t]
\centering
\caption{\mbox{SST-2} accuracy with different $d$ and $\epsilon$ for SeqLDP}
\begin{tabular}{r||r|r|r|r|r}
\toprule
\diagbox{$d$}{$\epsilon$} & $2$ & $4$ & $8$ & $12$ & $16$ \\
\midrule
$16$ & $0.6801$ & $0.7851$ & $0.8833$ & $\mathbf{0.9232}$ & $0.9266$ \\
$64$ & $0.6766$ & $0.7752$ & $0.8727$ & $0.9037$ & $0.9209$ \\
$128$ & $0.6732$ & $\mathbf{0.7856}$ & $0.8807$ & $0.9128$ & $0.9232$ \\
$256$ & $\mathbf{0.6835}$ & $0.7695$ & $\mathbf{0.8965}$ & $0.9186$ & $\mathbf{0.9249}$ \\
$512$ & $0.6411$ & $0.7626$ & $0.8831$ & $0.9128$ & $0.9243$ \\
$768$ & $0.6686$ & $0.7741$ & $0.8739$ & $0.9128$ & $0.9186$ \\
\bottomrule
\end{tabular}
\label{tbl:result_n}
\end{table}

\begin{table}[t]
\centering
\caption{Accuracy on output embeddings under SeqLDP}
\label{tab:cmp_dp_mech.}
\resizebox{\linewidth}{!}{
\begin{tabular}{l|l||c|c|c|c|c}
\toprule
Task & Mech. & $\epsilon = 0.5$ & $\epsilon = 1$ & $\epsilon = 2$ & $\epsilon = 4$ & $\epsilon = 8$ \\
\midrule
\multirow{3}{*}{\mbox{SST-2}} & GM & $0.5424$ & $0.5757$ & $0.6537$ & $0.7466$ & $0.8624$ \\ 
\cline{2-7}
 & aMGM & $\mathbf{0.5516}$ & $\mathbf{0.6021}$ & $\mathbf{0.6686}$ & $\mathbf{0.7741}$ & $\mathbf{0.8739}$ \\
 \cline{2-7}
 & MVG & \multicolumn{5}{c}{${\sim}0.5$} \\
\midrule
\multirow{3}{*}{\mbox{IMDb}} & GM & $0.5244$ & $0.5498$ & $0.6016$ & $0.6902$ & $0.8002$ \\ 
\cline{2-7}
 & aMGM & $\mathbf{0.5353}$ & $\mathbf{0.5676}$ & $\mathbf{0.6224}$ & $\mathbf{0.7093}$ & $\mathbf{0.8109}$ \\
 \cline{2-7}
 & MVG & \multicolumn{5}{c}{${\sim}0.5$} \\
\midrule
\multirow{3}{*}{\mbox{QQP}} & GM & $0.6304$ & $0.6321$ & $0.6571$ & $0.7458$ & $0.8638$ \\ 
\cline{2-7}
 & aMGM & $\mathbf{0.6312}$ & $\mathbf{0.6348}$ & $\mathbf{0.6747}$ & $\mathbf{0.7685}$ & $\mathbf{0.8653}$ \\
 \cline{2-7}
 & MVG & \multicolumn{5}{c}{${\sim}0.5$} \\
\bottomrule 
\end{tabular}
}
\end{table}

\begin{table*}[t]
\centering
\caption{Accuracy of task models fine-tuned using DP-Forward and DP-SGD under (example-level) CDP}
\label{tab:cmp_sgd_forward}
\vspace{-3pt}
\begin{tabular}{ll||ccc|ccc|ccc}
\toprule
\multicolumn{1}{l|}{\multirow{2}{*}{Method}} & \multirow{2}{*}{\begin{tabular}[c]{@{}l@{}}Noise \\ position\end{tabular}} & \multicolumn{3}{c|}{SST2} & \multicolumn{3}{c|}{IMDb} & \multicolumn{3}{c}{QQP} \\ \cline{3-11} 
\multicolumn{1}{l|}{} & & $\epsilon \approx 1$ & $\epsilon \approx3$ & $\epsilon \approx 8$ & $\epsilon \approx 1$ & $\epsilon \approx 3$ & $\epsilon \approx 8$ & $\epsilon \approx 1$ & $\epsilon \approx 3$ & $\epsilon \approx 8$ \\
\midrule
\multicolumn{1}{l|}{\multirow{8}{*}{DP-Forward}} & Embedding & $0.6055$ & $0.6146$ & $0.6278$ & $0.5$ & $0.5$ & $0.5$ & $0.6534$ & $0.6589$ & $0.6594$ \\ \cline{2-11} 
\multicolumn{1}{l|}{} & Encoder 1 & $0.7971$ & $0.8096$ & $0.8073$ & $0.5000$ & $0.5016$ & $0.5022$ & $0.7857$ & $0.7885$ & $0.7906$ \\
	\multicolumn{1}{l|}{} & Encoder 3 & $0.8096$ & $0.8394$ & $0.8463$ & $0.7525$ & $0.7545$ & $0.7549$ & $\mathbf{0.8513}$ & $\mathbf{0.8585}$ & $\mathbf{0.8607}$ \\
\multicolumn{1}{l|}{} & Encoder 5 & $0.8544$ & $0.8658$ & $0.8716$ & $0.7642$ & $0.7709$ & $0.7719$ & $\mathbf{0.8698}$ & $\mathbf{0.8765}$ & $\mathbf{0.8806}$ \\
\multicolumn{1}{l|}{} & Encoder 7 & $0.8624$ & $\mathbf{0.8819}$ & $\mathbf{0.8872}$ & $0.7765$ & $\mathbf{0.7883}$ & $\mathbf{0.7924}$ & $\mathbf{0.8840}$ & $\mathbf{0.8887}$ & $\mathbf{0.8926}$ \\
\multicolumn{1}{l|}{} & Encoder 9 & $\mathbf{0.8945}$ & $\mathbf{0.8979}$ & $\mathbf{0.9002}$ & $\mathbf{0.7995}$ & $\mathbf{0.8105}$ & $\mathbf{0.8181}$ & $\mathbf{0.8895}$ & $\mathbf{0.8941}$ & $\mathbf{0.8955}$ \\
\multicolumn{1}{l|}{} & Encoder 11 & $\mathbf{0.8819}$ & $\mathbf{0.8968}$ & $\mathbf{0.8985}$ & $\mathbf{0.8042}$ & $\mathbf{0.8187}$ & $\mathbf{0.8265}$ & $\mathbf{0.8952}$ & $\mathbf{0.8997}$ & $\mathbf{0.9007}$ \\ \cline{2-11} 
\multicolumn{1}{l|}{} & Output & $\mathbf{0.8865}$ & $\mathbf{0.9009}$ & $\mathbf{0.9055}$ & $\mathbf{0.8096}$ & $\mathbf{0.8160}$ & $\mathbf{0.8270}$ & $\mathbf{0.8987}$ & $\mathbf{0.8994}$ & $\mathbf{0.9038}$ \\ 
\hline
\multicolumn{1}{l|}{DP-SGD} & Gradient & $0.8650$ & $0.8713$ & $0.8759$ & $0.7779$ & $0.7826$ &$0.7903$ & $0.8219$ & $0.8345$ & $0.8433$ \\
\hline
\multicolumn{2}{l||}{Non-private baseline} & \multicolumn{3}{c|}{$0.9178$} & \multicolumn{3}{c|}{$0.8378$} & \multicolumn{3}{c}{$0.9019$} \\
\bottomrule
\end{tabular}
\vspace{-7pt}
\end{table*}

\subsection{Fine-tuning with Sequence LDP}
\label{sect:different_mech}
Our approach also supports perturbing \emph{sub-layer} outputs during fine-tuning.
We study six encoders as an example, with the results shown in Figure~\ref{fig:inside_encoders}.
Overall, DP-Forward performs better with deeper encoders since fewer parameters are directly affected by noise during fine-tuning. 
Another observation is that perturbing different sub-layer outputs, even inside the same encoder, may result in huge accuracy variation; \eg, using noisy outputs of the last sub-layer in Encoder~$1$ can bring ${\sim}20$pp gains over those of the first sub-layer.

We next evaluate the privacy-accuracy tradeoffs under different~$\epsilon$ and compare the instantiations using the classical GM, MVG~\cite{ccs/ChanyaswadDPM18}, and aMGM. 
Note that we still compute the GM variance as $\sigma^2 = 2 \ln (1.25/\delta)S^2_2(f) / \epsilon^2$ for empirical evaluation, albeit it cannot extend to $\epsilon > 1$ for a single run to ensure theoretical DP guarantees.

For the GM- and aMGM-based instantiations, Table~\ref{tab:cmp_dp_mech.} shows all three tasks' accuracy increases with $\epsilon$.
Ours has better accuracy than the GM-based one due to the smaller noise produced by aMGM in all choices of $\epsilon$. 
Although the noise variance gap (between GM and aMGM) widens as $\epsilon$ decreases, one cannot fine-tune effective models in a high privacy regime $\epsilon < 1$.
The MVG-based one behaves like random guessing for all three tasks since its noise variance is proportional to $n\cdot d$, which is even much larger than the classical GM for high-dimensional settings (see Section~\ref{sect:mvg}).
For instance, under the same parameter setting (\eg, $n = 128, d = 768$, and $\epsilon = 8$), MVG produces noise with the variance orders-of-magnitude larger than aMGM (\eg, ${>}10^8$ vs. ${\sim}0.6$), even assuming $\sup||f(\cdot)||_F = 1$.

We remark that the used local $\epsilon$ value is not large.
Most classical LDP works that deem such $\epsilon$ lies in a low privacy regime are for statistical analytics.
In great contrast, we aim at fine-tuning large LM-based pipelines with high-dimensional signals and limited training data, which is much more complicated.
Many prior works~\cite{wsdm/FeyisetanBDD20, cikm/QuKY0BN21, acl/YueDWLSC21, icdm/FeyisetanDD19} use a larger $\epsilon$ to ensure even a weaker token-level LDP variant, while others~\cite{acl/MeehanMC22} categorizes $\epsilon < 10$ and $10 \leq \epsilon < 20$ as strong and moderate privacy respectively\footnote{Such choices
can be ``reduced'' to smaller ones under the shuffling model (Section~\ref{exp:cmp}), \cf. U.S. census discloses demographic data at \emph{central} $\epsilon = 11.14$~\cite{corr/UScensus22}.} for sequence-level LDP like ours.
More importantly, they provide effective protection against various privacy threats, as detailed in Section~\ref{sect:attacks}.

\subsection{DP-Forward versus DP-SGD}
\label{exp:cmp}
\emph{Fairness of comparisons on privacy-accuracy tradeoffs.}
As elaborated in Section~\ref{sect:cmp_dpsgd}, we can adopt RR~\cite{jasa/Warner65} to perturb the labels and then report central $\epsilon$ values for DP-Forward, amplified by shuffling using the following parameters, ensuring that comparisons are fair under (example-level) CDP.
For DP-SGD, the subsampling probability is $b / N$, with $b=32$ and the dataset size $N$; the number of fine-tuning steps is $T = k \cdot N /b$ with $k = 3$.
For DP-Forward, the subsampling and non-flipping probabilities are respectively $1/N$ (with $T= k \cdot N$) and $0.9$; we still process $b$ noisy embeddings as a batch.
For both, we use aGM~\cite{icml/BalleW18}, the degenerated version of aMGM (Section~\ref{sect:agm-aMGM}), and the same accountant~\cite{nips/GopiLW21} to report \emph{approximated} overall~$\epsilon$ values\footnote{They are dominated by composing subsampled Gaussian, \eg, composing subsampled RR only consumes $0.03$ for \mbox{SST-2}, which is even overestimated by AutoDP.}.

We study eight instances of DP-Forward, including perturbing the outputs of the input embedding layer, six different encoders, and BERT.
Their accuracies on all three tasks under three privacy levels, plus those of DP-SGD and the non-private baseline, are shown in Table~\ref{tab:cmp_sgd_forward}.
About half or more of our instances have better accuracy than DP-SGD for each task; the largest accuracy gain is ${\sim}7.7$pp for QQP.
The noisy output embeddings often lead to the best accuracy for all tasks, even comparable to the non-private baseline, due to the dimension reduction at the last encoder output (Section~\ref{sect:pre_BERT}).

Recent DP-SGD variants~\cite{icml/YuZC0L21,iclr/YuNBGIKKL22} improve DP-SGD~\cite{ccs/AbadiCGMMT016} by perturbing \emph{partial gradient} entries using additional tricks (\eg, low-rank adaption).
They report the best accuracy of $92.5\%$ and~$85.7\%$ on \mbox{SST-2} and QQP, respectively, with $2.3$pp and $6.2$pp drops from~the non-private baselines at central $\epsilon = 6.7$~\cite[Table~$4$]{iclr/YuNBGIKKL22}.
DP-Forward with label privacy, incurring ${<}1.7$pp accuracy drops on the two tasks at $\epsilon \approx 3$, can still beat them, albeit their fine-tuning is based on RoBERTa-base, a robustly optimized BERT approach, which by itself outperforms BERT due to larger training set, longer training time, and better techniques (\eg, dynamic masking in MLM).

Figure~\ref{fig:cmp_dpsgd} shows the efficiency comparisons on fine-tuning \mbox{SST-2}.
The time and storage overheads of our approach (for all possible instances) are almost the same as the non-private baseline and ${\sim}3\times$ smaller than DP-SGD.
It is because we allow batch processing as in the normal fine-tuning -- no need to handle per-example~gradients.
Meanwhile, our normalization and noise sampling/addition are also faster since the size of embeddings is smaller than that of gradients.

\subsection{Noisy Pre-training}
Pre-training BERT using DP-Forward, aligned with the noisy fine-tuning, does help accuracy.
We use \mbox{SST-2} as an example and perturb the input embedding matrices.
We continue pre-training BERT over English WikiCorpus, the $2006$ dump with about $600$M words, for an epoch.
Table~\ref{tbl:noisypretrain} shows that we can obtain $1{-}2$pp accuracy gains for most choices of $\epsilon$, compared to fine-tuning on the original BERT.

Efficiency-wise, DP-Forward pre-training also consumes much fewer resources; \eg, an existing work~\cite{emnlp/AnilG00M22} pre-trains BERT-Large (with $340$ million parameters) using DP-SGD on Google TPUs,~which requires sufficient memory for handling batch sizes of millions.

\begin{figure}
	\centering
	\includegraphics[width=\linewidth]{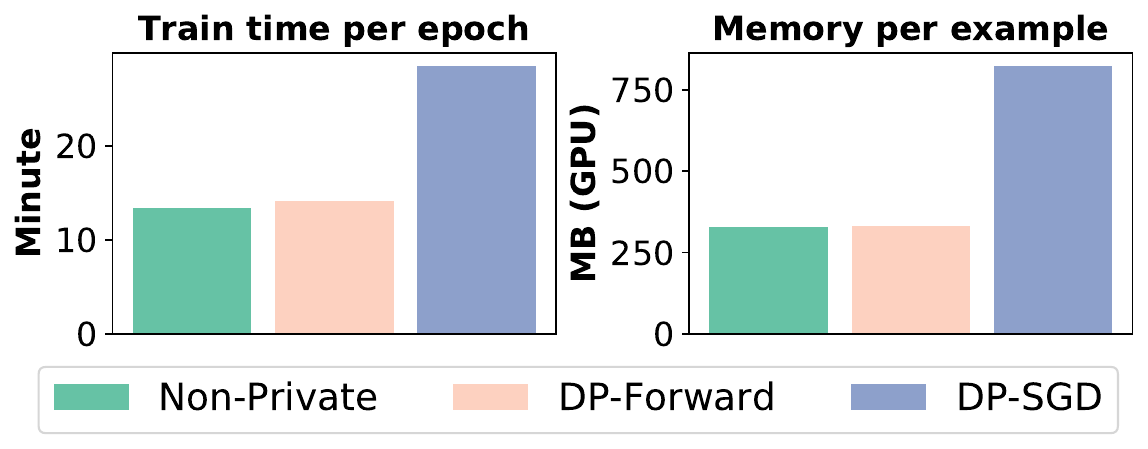}
	\vspace{-20pt}
	\caption{Efficiency comparison for the case of \mbox{SST-2}}
	\label{fig:cmp_dpsgd}
\end{figure}

\begin{table}[!t]
\centering
\vspace{1pt}
\caption{\mbox{SST-2} accuracy gain with pre-training under SeqLDP}
\vspace{-2pt}
\label{tbl:noisypretrain}
\begin{tabular}{r||r|r|r}
\toprule
$\epsilon$ & Raw BERT & Noisy BERT & $\Delta_{\text{acc}}$ \\
\midrule
$2$	 & $0.5501$ & $0.5665$	& $0.0164$ \\
$4$	 & $0.5950$ & $0.5999$	& $0.0049$ \\ 
$8$	 & $0.6550$ & $0.6708$	& $0.0158$ \\
$16$ & $0.7345$ & $0.7450$	& $0.0105$ \\
\bottomrule
\end{tabular}
\vspace{-3pt}
\end{table}

\section{Defense against Privacy Threats}
\label{sect:attacks}
Following the recent taxonomy~\cite{ccs/SongR20}, we study MIAs and two new threats of sequence leakage from their \emph{embeddings}: embedding inversion and attribute inference.
We moderately adapt them to suit our context, \eg, upgrading MIAs~\cite{ccs/SongR20} to sequence-level.

\subsection{Threat Models}
For MIAs, we follow prior arts~\cite{sp/ShokriSSS17,csfw/YeomGFJ18,ccs/SongR20} to consider an adversary with only \emph{black-box} access to an \emph{entire} (DP-SGD/DP-Forward-trained) pipeline: It can query the prediction results (\eg, each-class probability) of target sequences but cannot access the pipeline weights and architecture; the hidden embeddings are not revealed.

Despite ``different'' objectives of inverting or inferring (Section~\ref{sect:inversion} or~\ref{sect:sensitive_attribute}) from embeddings, we consider both threats involving a general adversary with \emph{black-box} access to the trained pipeline part~$f()$.
It can get the inference-time (clear/noisy) \emph{embeddings} of target sequences~\cite{ccs/SongR20}.
Besides public prior knowledge, it can collect a \emph{limited auxiliary} dataset $\D_{\aux}$, sharing similar attributes to the targets~\cite{ccs/SongR20}.

DP-SGD only offers CDP for training data and does not protect inference-time input.\footnote{One might add the same noise to it as DP-Forward inference, which indeed mitigates the new threats.
However, perturbing \emph{gradients} in training, inherently ``mismatches'' from perturbing \emph{embeddings} in inference, deteriorating task performance significantly, \eg, \mbox{SST-2} accuracy will be reduced to $0.7786$ (with a ${\sim}10$pp drop) at central $\epsilon \approx 8$.}
What follows intends to empirically confirm a major merit of DP-Forward in protecting against stronger adversaries and threats to both training- and inference-time inputs.

\subsection{Membership Inference Attacks}
\label{sect:MIA}

\smallskip
\noindent\textbf{Attack Objective.}
MIAs predict whether a data point is in the~training set~\cite{sp/ShokriSSS17}.
They often exploit the disparity in model behavior between training data and unseen data, \ie, poor model generalization due to overfitting~\cite{csfw/YeomGFJ18}.
Inferring membership at the token/word level, \eg, a sliding window of tokens~\cite{ccs/SongR20}, is not interesting.
We consider more realistic MIAs on \emph{entire} sequences, which can be extended for more devastating attacks, such as extracting verbatim pre-training sequences via black-box access to \mbox{GPT-2}~\cite{uss/CarliniTWJHLRBS21}.

Prior arts~\cite{csfw/YeomGFJ18,uss/SongM21} suggest that threshold-based MIAs using only prediction confidence~\cite{csfw/YeomGFJ18} or entropy~\cite{uss/SongM21} with proper assumptions are comparable to the more sophisticated one~\cite{sp/ShokriSSS17} based on shadow training.
Adapting the confidence-based MIA to our context exploits that a pipeline is fine-tuned by minimizing its prediction loss: The confidence/probability of predicting a training sequence as its true label should be close to~$1$.
The adversary can then infer a candidate sequence $X^*$ as a member when the confidence for the predicted label $l$ output by pipeline~$\mathcal{F}$ is larger than a pre-set threshold $\tau$:
$$\mathds{1}\{\Pr[\mathcal{F}(X^*) = l] \geq \tau \},$$
where $\mathds{1}\{\cdot\}$ is the indicator function.
We simply use a fixed~$\tau$ for all possible labels in our evaluation, albeit it can be label-dependent.

The second MIA we use is based on the prediction output (\ie, a vector of probabilities) of a training sequence tends to be a one-hot vector, \ie, its entropy should be close to $0$.
Similarly, the adversary can infer $X^*$ as a member when its prediction entropy falls below a preset threshold $\tau$; otherwise, it is not deemed a member:
$$\mathds{1}\{- \sum_i \Pr[\mathcal{F}(X^*)= l_i] \log(\Pr[\mathcal{F}(X^*)= l_i]) \leq \tau \},$$
for all possible labels $\{l_i\}$.
Note that a totally wrong prediction with probability ${\sim}1$ also leads to entropy approaching~$0$.
We can address it by encoding 
the information of the ground-truth label of $X^*$~\cite{uss/SongM21}.

\smallskip
\noindent\textbf{Numerical Results.}
As in~\cite{icml/YuZC0L21}, all the test examples 
and a random subset of the training examples (as many as the test ones) are evenly split into two subsets (each has half of the training/test examples), one for finding the optimal $\tau$, and the other for reporting the attack success rates.
Given that the training and test examples likely share the same distribution, we randomly drop/replace tokens in the test examples to enlarge the prediction difference to make MIAs easier.

We evaluated the adapted confidence- and entropy-based MIAs on SST-2 fine-tuned by the non-private baseline, DP-Forward, and DP-SGD.
For DP-Forward, we investigate five instances, perturbing input embeddings, three encoders' outputs, and output embeddings.
Table~\ref{tbl:seq_MIA} presents the results, where success rates within $0.49${--}$0.51$ are shown in bold.
Both DP-Forward and DP-SGD can mitigate MIAs effectively.
For all choices of $\epsilon$, the two MIAs' success rates on DP-Forward are reduced to ${\sim}0.5$ (like random guessing) for deeper layers, outperforming DP-SGD by ${>}6$pp at the same privacy level.

\begin{table}[!t]
\centering
\caption{Success rates of two MIAs with (translated) central~$\epsilon$}
\label{tbl:seq_MIA}
\vspace{-2pt}
\begin{tabular}{c|l||cc}
\toprule
\multirow{2}{*}{$\epsilon$} & \multirow{2}{*}{Method} 
& \multicolumn{2}{c}{Attack Success Rate} \\ \cline{3-4} 
 & 
 & Entropy & Confidence \\
 \midrule
 $\infty$ & Non-private baseline
 & $0.659$ & $0.645$ \\
 \midrule
\multirow{6}{*}{1} & DP-SGD 
& $0.567$ & $0.561$ \\ \cline{2-4} 
 & DP-Forward (Embedding) 
& $0.586$ & $0.576$ \\ 
 & DP-Forward (Encoder 1) 
& $0.535$ & $0.537$ \\ 
 & DP-Forward (Encoder 7) 
 & $\mathbf{0.494}$ & $\mathbf{0.506}$ \\ 
 & DP-Forward (Encoder 11) 
 & $\mathbf{0.506}$ & $\mathbf{0.494}$ \\ 
 & DP-Forward (Output) 
 & $\mathbf{0.508}$ & $\mathbf{0.502}$ \\
 \midrule
\multirow{6}{*}{3} & DP-SGD 
& $0.584$ & $0.576$ \\ \cline{2-4} 
 & DP-Forward (Embedding) 
& $0.584$ & $0.576$ \\ 
 & DP-Forward (Encoder 1) 
& $0.543$ & $0.530$ \\ 
 & DP-Forward (Encoder 7) 
& $\mathbf{0.510}$ & $\mathbf{0.507}$ \\ 
 & DP-Forward (Encoder 11) 
 & $0.512$ & $\mathbf{0.500}$ \\ 
 & DP-Forward (Output) 
 & $\mathbf{0.503}$ & $\mathbf{0.499}$ \\
 \midrule
\multirow{6}{*}{8} & DP-SGD 
& $0.580$ & $0.580$ \\ \cline{2-4} 
 & DP-Forward (Embedding) 
& $0.597$ & $0.576$ \\
 & DP-Forward (Encoder 1) 
& $\mathbf{0.510}$ & $0.536$ \\ 
 & DP-Forward (Encoder 7) 
& $\mathbf{0.510}$ & $\mathbf{0.504}$ \\ 
 & DP-Forward (Encoder 11) 
 & $0.520$ & $\mathbf{0.513}$ \\ 
 & DP-Forward (Output) 
 & $\mathbf{0.506}$ & $\mathbf{0.490}$ \\
 \bottomrule
\end{tabular}
\vspace{-6pt}
\end{table}

\subsection{Embedding Inversion Attacks}
\label{sect:inversion}
\noindent\textbf{Attack Objective.}
These attacks aim at recovering the raw text as (unordered) tokens $\{x_i\}_{i \in [n]} \subseteq X$ from embeddings, highlighting the risk of directly sharing (without noise) even only text embeddings (for training/inference).
They have been employed to reconstruct specific patterns, \eg, identity codes and gene segments~\cite{sp/PanZJY20}.

We first propose a simple \emph{token-wise} inversion attack to invert (noisy) token embeddings output by the input embedding layer~$\phi(\cdot)$ that maps every token in $\V$ to $\mathbb{R}^d$~\cite{corr/WuSCLNMKCGMKSJL16}.
It can be formulated as:
$$\forall i \in [n]: \min_{x_i^* \in \V} ||\phi(x^*_i)-(\phi(x_i)+z_i)||_2,$$
where $z_i$ is the $i^{\text{th}}$ row of noise $Z$ from $\M$ (omitted for DP-SGD or the non-private baseline).
It returns $x^*_i$ with its embedding closest to the observed one of $x_i$ via a nearest-neighbor search over~$\V$.

A token's hidden embedding from deeper layers encodes more ``abstract'' contextual information of the entire sequence it belongs to; the token-wise inversion may be less accurate.
We thus require a more general attack~\cite{ccs/SongR20}.
It first maps the observed (noisy) embedding back to a lower-layer one using a linear least square model~$M$ and then selects $n$ tokens as $X^*$ to minimize the $L_2$-distance between the lower-layer representation of $X^*$ and the one from $M$:
$$\min_{X^* \in \V^n}||\zeta(X^*)-M\left(f(X)+Z\right)||_2,$$ 
where $\zeta(\cdot)$ is a lower-layer representation function than $f(\cdot)$.

The above minimization is over $|\V|^n$, larger than the token-wise candidate space.
To determine $X^*$, we first relax the token selection at each position $i \in [n]$ using a continuous vector in $\mathbb{R}^{|\V|}$, which is then input (with another temperature parameter) to a softmax function to model the probabilities of selecting each token in $\V$.
We further derive the token embedding $x_i^*$ by multiplying the relaxed vector (with each entry as a weight) and the original embedding matrix.
Finally, we solve it by a gradient-based method~\cite{ccs/SongR20}.

\smallskip
\noindent\textbf{Numerical Results.}
The gradient-based attack reports the highest recall (or precision) on inverting the lowest-layer (clear) embeddings~\cite[Figure~2]{ccs/SongR20}.
To show that DP-Forward can mitigate such ``strongest'' inversion, we implement both (nearest-neighbor and gradient-based) attacks to invert input embeddings, with the public BERT embedding lookup table as prior.
We also report their success rates as \emph{recall} -- the ratios of correct recoveries over the raw targets.

Table~\ref{tbl:embedding_inversion} shows that DP-Forward can reduce their success rates to a relatively low level, most are within $0.2$.
However, DP-SGD fails in defense.
The results corroborate our claim: DP-Forward directly adds noise to embeddings, thus mitigating embedding inversion, whereas DP-SGD only perturbs gradients, offering no protection for the (clear) inference-time embeddings of test sequences.

\begin{table}[!t]
\centering
\caption{Success rates of two inversion attacks on (the lowest-layer) input embeddings with (translated) central~$\epsilon \approx 8$}
\label{tbl:embedding_inversion}
\vspace{-2pt}
\begin{tabular}{l||r|r|r||r|r|r}
\toprule
& \multicolumn{3}{c||}{Nearest Neighbor} & \multicolumn{3}{c}{Gradient-based} \\
\cline{2-7}
& SST-2 & IMDb & QQP & SST-2 & IMDb & QQP \\
\midrule
\hspace{-8pt} Non-private \!\!\!& 
$1$\!\! & 
$1$\!\! & 
$1$\!\! & 
$1$\!\! & 
$1$\!\! & 
$1$\!\! \\
\hspace{-8pt} DP-SGD \!\!\!& 
$1$\!\! & 
$1$\!\! & 
$1$\!\! & 
$.9991$\!\! & 
$.9982$\!\! & 
$1$\!\! \\
\hspace{-8pt} DP-Forward\!\! & 
\!\!$\mathbf{.1811}$\!\! & 
\!\!$\mathbf{.1420}$\!\! & 
\!\!$\mathbf{.2457}$\!\! &
\!\!$\mathbf{.1622}$\!\! &
\!\!$\mathbf{.1241}$\!\! &
\!\!$\mathbf{.2226}$\!\! \\
\bottomrule
\end{tabular}
\end{table}

\subsection{Sensitive Attribute Inference Attacks}
\label{sect:sensitive_attribute}
\noindent\textbf{Attack Objective.}
Instead of recovering exact tokens, one can try to infer sensitive attributes about target sequences from their embeddings.
The attributes are often statistically unrelated to the training/inference objective but inherent in sequences, \eg,~\mbox{stylometry}, implying the text's authorship for sentiment analysis~\cite{uss/ShettySF18}.
We~are not interested in any global property of an entire corpus~\cite{ccs/GanjuWYGB18}.

\begin{table}[!t]
\centering
\caption{Success rates of a (neural-network-based) sensitive attribute inference attack with (translated) central~$\epsilon \approx 8$}
\label{tbl:sensitive_attributes}
\vspace{-2pt}
\resizebox{\linewidth}{!}{%
\begin{tabular}{l||r|r|r|r||r}
\toprule
 & action & comedy & drama & horror & Overall\! \\
 \midrule
\hspace{-4pt}Non-private & $0.727$ & $0.858$ & $0.516$ & $0.439$ & $0.687$\! \\ 
\midrule
\hspace{-4pt}DP-SGD &$0.664$ & $0.733$ & $0.253$ & $0.324$ & $0.536$\! \\
\midrule
\hspace{-4pt}DP-Forward (Embedding)\!\! & $0.998$ &$\mathbf{0}$\! &$\mathbf{0}$ &$\mathbf{0.009}$ &$\mathbf{0.276}$\! \\ 
\hspace{-4pt}DP-Forward (Encoder 1)\!\! & $1.0$ &$\mathbf{0}$\! &$\mathbf{0}$ &$\mathbf{0}$ &$\mathbf{0.276}$\! \\ 
\hspace{-4pt}DP-Forward (Encoder 7)\!\! & $1.0$ &$\mathbf{0}$\! &$\mathbf{0}$ &$\mathbf{0}$ &$\mathbf{0.276}$\! \\ 
\hspace{-4pt}DP-Forward (Encoder 11)\!\! & $1.0$ &$\mathbf{0}$\! &$\mathbf{0}$ &$\mathbf{0}$ &$\mathbf{0.276}$\! \\ 
\hspace{-4pt}DP-Forward (Output)\!\! & $0.998$ &$\mathbf{0}$\! &$\mathbf{0}$ &$\mathbf{0.009}$ &$\mathbf{0.276}$\! \\
\bottomrule
\end{tabular}
}
\vspace{-8pt}
\end{table}

We assume $\D_{\aux}$ has sequences labeled with sensitive attributes.
The adversary can query $f(\cdot)$ for the noisy (or clear) embeddings:
$$\D_{\aux} = \{\left(f(X_i)+Z_i,s_i\right)\},\ \forall s_i \in \mathcal{S},$$
where $\mathcal{S}$ is the set of all possible sensitive attributes of interest, say, authorship.
It does not care about non-sensitive attributes.

To infer sensitive attributes, the adversary first trains a classifier on $\D_{\aux}$ via supervised learning and then uses it for an observed noisy (or clear) embedding $f(X^*)+Z$ to predict $s^* \in \mathcal{S}$ of~$X^*$.

\smallskip
\noindent\textbf{Numerical Results.} 
We train a three-layer neural network with a linear head as the classifier to infer the film genre (\eg, `horror') as a sensitive attribute from a movie review using its output embedding.
We employ IMDb with $20$k examples ($90\%$ for training and $10\%$ for validation) as $\D_\aux$, and \mbox{SST-2} contributes $3.3$k examples for testing the classifier.
The attack success rates are measured using \emph{recall}.

We investigate five DP-Forward instances.
Table~\ref{tbl:sensitive_attributes} shows that they ``reduce'' the classifier to majority-class prediction, which returns the majority class (`action') on all inputs.
In contrast, DP-SGD only reduces success rates moderately compared to the non-private baseline.
It is because the embeddings from DP-SGD-trained/noisy models still ``lose'' some useful information (\cf, accuracy drops~of DP-SGD inference on embeddings without noise).
The results confirm DP-Forward is more effective in thwarting attribute inference.

\section{Related Work}

\subsection{Privacy Threats on LMs and Embeddings}
An active line of research~\cite{sp/PanZJY20,ccs/SongR20,ccs/BeguelinWTRPOKB20,uss/CarliniTWJHLRBS21} discloses severe privacy risks in modern LMs (even used as black-box query ``oracles'') concerning their (hidden/output) text embeddings.
Song and Raghunathan~\cite{ccs/SongR20} build a taxonomy of attacks that covers a broader scope than a parallel work~\cite{sp/PanZJY20}.
These attacks include embedding inversion (which can partially recover raw texts), membership inference (establishing the \emph{is-in} relation between a target and private training data), and inferring sensitive attributes like text authorship from embeddings.
A common defense for them is adversarial training, \eg,~\cite{emnlp/ElazarG18}.

Others~\cite{ccs/BeguelinWTRPOKB20,uss/CarliniTWJHLRBS21} study the ``memorization'' of training data in LMs (a.k.a. membership inference attack).
In particular, Carlini~\etal~\cite{uss/CarliniTWJHLRBS21} define $k$-eidetic memorization, where a string is extractable or memorized if it appears in at most $k$ examples.
Their black-box attacks on GPT-2~\cite{report/RadfordNSS18} can extract verbatim training texts even when $k = 1$ (\eg, a name that only appears once is still extractable).
A smaller~$k$ means a higher privacy risk.
Beguelin~\etal~\cite{ccs/BeguelinWTRPOKB20} define differential score and rank as two new metrics for analyzing the update leakage, enabling the recovery of new text used to update LMs.
Incorporating DP to address memorization is a promising~solution.

\subsection{Input (Text/Feature) Perturbation for LDP}
SynTF~\cite{sigir/WeggenmannK18} synthesizes
term-frequency (feature) vectors under LDP, which have limited applications compared to sentence embeddings or text itself.
Feyisetan~\etal~\cite{icdm/FeyisetanDD19,wsdm/FeyisetanBDD20} resort to metric-LDP~\cite{csfw/Alvim0PP18}, a relaxed variant of LDP with a distance metric~(\eg, Euclidean or Hyperbolic), which allows the indistinguishability of outputs to grow proportionally to the inputs' distance.
They first add noise to the outputs of a non-contextualized token embedding model (\eg, GLoVe~\cite{emnlp/PenningtonSM14}), which are then projected back to ``sanitized'' text using the nearest neighbor search as post-processing.
In contrast, Yue~\etal~\cite{acl/YueDWLSC21} 
sanitize text by directly sampling token-wise replacements, avoiding adding noise to high-dimensional embeddings.
All these works only achieve (variants of) token-level metric-LDP.

To offer \emph{sequence-level} protection, recent studies~\cite{emnlp/LyuHL20, acl/MeehanMC22} apply Laplace or exponential mechanism to perturb (the average of) sentence embeddings extracted by an LM (\eg, BERT~\cite{naacl/DevlinCLT19}).
Both ensure pure LDP (homogeneously protecting any entire sequence), which may be too stringent and impact utility.
In contrast, heterogeneous protection~\cite{wsdm/FeyisetanBDD20,acl/YueDWLSC21} can strategically manage the privacy demands across inputs.
Du~\etal~\cite{www/DuYCS23} achieve metric-LDP (by Purkayastha and planar Laplace mechanisms) at the sequence level (unlike token-level in prior arts~\cite{wsdm/FeyisetanBDD20,acl/YueDWLSC21}).
To further boost the utility, they mitigate the dimensional curse via a random-projection-like approach.
They also perturb sensitive sequence labels for enhanced privacy.
Nevertheless, perturbing different hidden (rather than token or sentence) embeddings inside LM-based NLP pipelines remains unexplored.

\subsection{DP-SGD (Variants) in Training LMs}
\label{sect:related_dp-sgd}
An early attempt~\cite{iclr/McMahanRT018} uses DP-SGD to train long short-term memory LMs in the federated learning setting.
By configuring hyperparameters properly (\eg, setting the batch size to millions), one can even pre-train BERT-Large, an LM with ${\sim}340$M parameters, using DP-SGD/Adam while achieving acceptable (MLM) accuracy~\cite{emnlp/AnilG00M22}.

Using the vanilla DP-SGD in pre-training/fine-tuning large LMs leads to significant efficiency and accuracy drops due to the ``curse of dimensionality.''
Yu~\etal~\cite{icml/YuZC0L21} propose reparametrized gradient perturbation:
It first reparameterizes/decomposes \emph{each} high-rank weight matrix into two low-rank (gradient-carrier) ones with a residual matrix and then only perturbs the two low-rank gradients to alleviate the dimensional curse.
The noisy low-rank gradients are finally projected back to update the raw high-rank weights.

Applying reparameterization to every weight in each update is still costly and may introduce instability (\eg, noises are ``zoomed~up'' during the projection).
Instead, the follow-up~\cite{iclr/YuNBGIKKL22} builds atop the recent success of parameter-efficient fine-tuning (\eg, LoRA~\cite{iclr/HuSWALWWC22}, Adapter~\cite{icml/HoulsbyGJMLGAG19}, and Compacter~\cite{nips/MahabadiHR21}): It perturbs the gradients of a much smaller number of additional ``plug-in'' parameters.
However, Li~\etal~\cite{iclr/LiTLH22} empirically show that parameter-efficient fine-tuning is not necessarily better than the full one;
they propose ghost clipping, a memory-saving technique (``orthogonal'' to dimension reduction), to use DP-SGD in full fine-tuning without instantiating per-example gradients.
Despite efficiency/accuracy gains, all these works still only protect training data by perturbing (smaller) gradients.

\subsection{DP Mechanisms for Matrix Functions}
Gaussian and Laplace mechanisms are typically for scalar-/vector-valued functions~\cite{fttcs/DworkR14}.
Vectorizing the outputs and adding i.i.d.~noise could generalize them for matrix-valued functions, but the structural information of matrix functions is not exploited.
The MVG mechanism~\cite{ccs/ChanyaswadDPM18} is thus devised, which draws directional or {non-i.i.d.} noise
from a \emph{matrix} Gaussian distribution.
It injects less noise into more ``informative'' output directions for better utility, with only a constraint on the sum of the singular values (determining the noise magnitude) of two covariance matrices.
Such a constraint is only a sufficient condition for $(\epsilon,\delta)$-DP, which is improved by the follow-up~\cite{tmc/YangXYWGLL23} with a tighter bound on the singular values.

There also exist mechanisms dedicated to \emph{restricted} matrix-valued functions.
The matrix mechanism~\cite{vldb/LiMHMR15} considers a collection of linear counting queries represented by $Wx$ for query matrix $W$ and input vector $x$.
It still resorts to additive Laplace/Gaussian noise but with an extra transformation solving the min-variance estimation to the noisy $Wx$.
Another very recent study~\cite{pvldb/Ji0YAYS21} focuses on matrix-valued queries with only binary (matrix) outputs.
It then devises an exclusive-or (xor) mechanism xor-ing the outputs with noise attributed to a matrix-valued Bernoulli distribution.

\section{Conclusion}
Pre-trained LMs became pivotal in NLP.
Alarmingly, fine-tuning corpora or inference-time inputs face various privacy attacks.
The popular DP-SGD only provides limited protection for training data by adding noise to gradients.
Raw tokens or sensitive attributes of training/inference data can be inverted or inferred from embeddings in forward-pass computation.
Vanilla DP-SGD also imposes high GPU memory and computational burdens but cannot be batched.

We propose DP-Forward, which directly adds noise to embedding matrices derived from the raw training/inference data in the forward pass.
Its core is the analytic matrix Gaussian mechanism, a general-purpose tool that owns independent interests.
It draws optimal matrix-valued noise from a matrix Gaussian distribution in a dedicated way using a necessary and sufficient condition for~DP.

Perturbing embeddings at various positions across multiple layers yields at least two benefits.
DP-Forward users are only required to download pipeline parts for deriving noisy embeddings, which is more storage- and time-efficient than deriving noisy gradients.
Together with our prior attempts~\cite{acl/YueDWLSC21,www/DuYCS23} at sanitizing input text tokens and output sentence embeddings, we provide a full suite of forward-pass signal sanitization options for users only to share their sanitized data for LM-as-a-Service APIs while protecting~privacy.

Beyond the theoretical contribution of two local DP notions and the experimental comparisons with baselines 
(\eg, GM, MVG, and DP-SGD) across three typical NLP tasks,
we investigate the hyperparameter configuration for reproducible validations of DP-Forward's potential in terms of efficiency, accuracy, and its ability to withstand diverse against diverse attacks.

Altogether, our new perspective leads to a better approach to privacy-aware deep neural network training, challenging the traditional wisdom focusing on gradients.
As a new paradigm for local DP in fine-tuning and inference, our work paves the way for a myriad of possibilities for new machine-learning privacy research~\cite{sp/NgC23}, \eg, generalization to transformer-based computer vision tasks.

\begin{acks}
We are grateful to the anonymous reviewers for their comments and Ashwin Machanavajjhala for his comments on a related Ph.D. thesis.
Sherman Chow is supported in part by the General Research Funds (CUHK 14209918, 14210319, 14210621), Research Grants Council, University Grants Committee, Hong Kong.
Authors at OSU are sponsored in part by NSF IIS $\#1815674$, NSF CAREER $\#1942980$, and Ohio Supercomputer Center \cite{OhioSupercomputerCenter1987}.
Tianhao Wang is supported in part by National Science Foundation (NSF) with grants CNS-$2220433$ and OAC-$2319988$.
\end{acks}

\bibliographystyle{ACM-Reference-Format}
\balance
\bibliography{reference}

\ifnum\fullFlag=1
\appendix

\section{Token-level DP-Forward}
\label{apdx:token-level}

\subsection{Definition and Related Notions}

\begin{definition}[Token-level SeqLDP]
\label{def:token-ldp}
\!\!For $\epsilon \geq 0, 0 \leq \delta \leq 1$, $\M$ fulfills token-level $(\epsilon,\delta)$-SeqLDP, if $\forall X \simeq X'$ that differ in any single token but with the same $y$, and any possible output subset $\O$, 
\begin{align*}
    \Pr[\M(X,y) \in \O] \leq e^\epsilon \Pr[\M(X',y) \in \O] + \delta.
\end{align*}
\end{definition}

Despite a token-level notion, 
our experiments (Appendix~\ref{apdx:exp}) show that
when $f(\cdot)$ is only the input embedding layer, 
our token-level SeqLDP designs can also effectively mitigate MIAs on \emph{entire} sequences,
with up to $20$pp accuracy gains at the same choices of $\epsilon$.
It is not necessarily weaker than sequence-level CDP (as offered by DP-SGD).
One might doubt its usefulness since two neighboring sequences may be too similar.
Nevertheless, there are cases where a sentence, \eg, ``How's it going'' may not matter in a bigger unit (paragraph/essay) of the training data either.
Moreover, a token (\eg, yes/no) can play a crucial role, \eg, in named entity recognition~\cite{tkde/LiSHL22}.
Our LDP guarantee is for \emph{any} such two sequences, covering the wide spectrum between ``too similar'' and radically different cases.

Note that weakening privacy notions by itself is not our goal\footnote{As a related example,
in image classification, PixelDP~\cite{sp/LecuyerAG0J19} has been proposed for a DP notion defined upon pixels.
Its motivation is robustness to adversarial examples.}.
Protection at the token level has been studied under metric-DP~\cite{wsdm/FeyisetanBDD20,cikm/QuKY0BN21}, a relaxation of LDP.
They require even much larger~$\epsilon$, say,~$175$.
Our goal of studying token-level SeqLDP is to narrow the gap~between theory and practice, \ie, 
provable privacy notions tailored to the protection targets
(the first few layers vs. the whole pipeline).

\subsection{Two Token-level SeqLDP Designs}
\label{sect:token-level}

For token-level SeqLDP, we need to bound a ``new'' $S_2(f), \forall X \simeq X'$, which should be tight and smaller than the one over $\forall X, X'$, hence producing smaller noise for better utility at meaningful \emph{token-level}~$\epsilon$.
It is still non-trivial since $f(\cdot)$, except for being the input embedding layer, may differ in every entry for even $X \simeq X'$.
One could also normalize the entire $f(\cdot)$ for $S_2(f),\forall X \simeq X'$, which ``degenerates'' to the token-level SeqLDP.
Instead, we tailor two designs to estimate a tighter $S_2(f)$ than the ``general'' one for only the input embedding layer and the first two layers, respectively.
Specifically, we employ \emph{row-wise} normalization and the Lipschitz continuity~\cite{icml/KimPM21}.

\smallskip
\noindent \textbf{After the Input Embedding Layer.}
When $f(\cdot)$ is only the input embedding layer, we work on each row $x_i$ independently: $||x_i||_2 = C, \forall i \in [n]$, where $||\cdot||_2$ is vector $2$-norm.
As a token only affects one row, we have $S_2(f) = C$, independent of whether the embedding layer will be updated.
Again with $\B$, we can draw noise~$Z \in \mathbb{R}^{n \times d}$.

\smallskip
\noindent \textbf{In the First MHA Sub-layer.}
The second option could be adding~$Z$ right after the first MHA sub-layer: $\MHA(X)+Z$, where $\MHA(\cdot)$~is the concatenation of $\Att_{i}(\cdot), i \in [h]$.
Yet, it is non-trivial to estimate $S_2(f)$ of $\MHA(\cdot)$ as $\Att_i(\cdot)$, let alone $\MHA(\cdot)$, is not Lipschitz~\cite{icml/KimPM21}.

\begin{definition}[Lipschitz Continuity]
Given two metric spaces $(\mathcal{X},d_{\mathcal{X}})$ and $(\mathcal{Y},d_{\mathcal{Y}})$, a function $f: \mathcal{X} \rightarrow \mathcal{Y}$ is Lipschitz continuous ($K$-Lipschitz) if there exists a constant $K \geq 0$,
$$d_\mathcal{Y}(f(X),f(X')) \leq K d_\mathcal{X}(X,X'), \forall X, X' \in \mathcal{X}.$$
The smallest $K$ is the Lipschitz constant, denoted by $\lip(f)$.
\end{definition}

We consider that $\mathcal{X}$ is the space of \emph{row-wise} normalized matrices in $\mathbb{R}^{n \times d}$, $\mathcal{Y}$ is the output space $\mathbb{R}^{n \times d(/h)}$ of $\Att_{i \in [h]}(\cdot)$ or $\MHA(\cdot)$, and $d_{\mathcal{X}} = d_{\mathcal{Y}} = ||\cdot||_F$ (or $p$-norm $||\cdot||_p$).
$\lip(f)$ generalizes $S_2(f)$ since it focuses on \emph{any} two inputs rather than just neighboring ones, allowing us to estimate an upper bound for $S_2(f)$ given $\lip(f)$.

The non-Lipschitz continuity stems from the non-linear Softmax activation, which takes pairwise dot products as input~\cite{icml/KimPM21}.
To make MHA Lipschitz, one might apply pairwise $L_2$-distances (hence called $L_2$-MHA)~\cite{icml/KimPM21} or add a normalization step called LipschitzNorm~\cite{icml/DasoulasSV21} in $\softmax(\cdot)$.
Unfortunately, estimating $\lip(f)$ of $L_2$-MHA needs to solve an intractable optimization problem, and LipschitzNorm is ill-suited for the high-dimensional BERT attention.

Instead of adding $Z$ to the outputs of $\MHA(\cdot)$ or $\Att_i(\cdot)$, we~can shift $f(\cdot)$ inside $\softmax(\cdot)$, where estimating $\lip(f)$ or $S_2(f)$ is feasible, \eg, the linear maps used to derive $Q,K,V$ matrices.
Considering a linear map $f(x) = xW$ with $W \in \mathbb{R}^{d \times d/h}$ and $x \in \mathbb{R}^{d}$, the $2$-norm $\lip_2(f)$ is the largest singular value $\sigma_{\max}(W)$~\cite{icml/KimPM21}.
When generalizing $f(\cdot)$ for any two matrices $X \simeq X'$, we can estimate 
$$S_2(f)= \sup ||f(X)-f(X')||_F =||f(x)||_2 \leq C \cdot \sigma_{\max}(W).$$
We can now respectively derive the noisy $Q, K, V$ matrices for $p(\cdot)$.
The first step is to draw noise $Z^{Q^*,K^*,V^*} \in \mathbb{R}^{n \times d}$ given $W^{Q^*,K^*,V^*} \in \mathbb{R}^{d \times d}$.
It requires us to either estimate $\sigma_{\max}(W^{Q^*, K^*, V^*})$ on the fly via power iteration or fix the linear maps in each forward pass of fine-tuning or inference.
We then compute $XW^{Q^*,K^*,V^*}+Z^{Q^*,K^*,V^*}$, which are reshaped into $3h$ matrices of size $n \times d/h$ for $\Att_{i \in [h]}(\cdot)$.

\begin{theorem}
\label{theo:token-ldp}
The two instances (with row-wise normalization) for fine-tuning or inference fulfill token-level $(\epsilon,\delta)$-(Seq)LDP.
\end{theorem}

The proof is equivalent to our approach for (Seq)LDP.
One just needs to compute $S_2(f), \forall X \simeq X'$ properly, and we did.

\paragraph{Discussion.}
For minimal changes to the pipeline, we adopt the~raw WordPiece~\cite{corr/WuSCLNMKCGMKSJL16}, which splits text into sub-words; using word-level tokenization yields word-level (Seq)LDP.
Our notion can also extend to \emph{phrase-level} (Seq)LDP by directly using the group privacy~\cite{fttcs/DworkR14} or dedicatedly computing the $L_2$-sensitivity smaller than $c \cdot S_2(f)$ for two sequences differing in (consecutive) $c$ tokens.
Typically, $c$ is small since a few tokens are enough for most sensitive information.
One could also add noise deeper in a pipeline using~$S_2(f_1 \circ f_2) \leq S_2(f_1) \cdot S_2(f_2)$, where $f_1 \circ f_2$ is function composition $f_1(f_2(\cdot))$.
We then need to estimate $S_2(f)$ of each (component~of) sub-layer.
For example, $\FFN(\cdot)$ has two linear maps $W_1$ and $W_2$ with $\ReLU(\cdot)$ in between, where $S_2(f)$ of $\ReLU(\cdot)$ is~$1$.
For $W_{1,2}$, its $S_2(f)$ is bounded by $\sqrt{d}C\sigma_{\max}(W_{1,2})$ since 
$||\cdot||_F \leq \sqrt{d}||\cdot||_2$ with $d$ as the rank.
We can also estimate $S_2(f)$ of $\LN(\cdot)$ from its Lipschitz constant~\cite{icml/KimPM21}.
When $f(\cdot)$ is composed of more layers, we can only get a looser estimation on the final $S_2(f)$.
Hence, our general recommendation is to add noise early when estimating a tight $S_2(f)$ is feasible.

\begin{figure*}[!t]
    \centering
    \includegraphics[width=\linewidth]{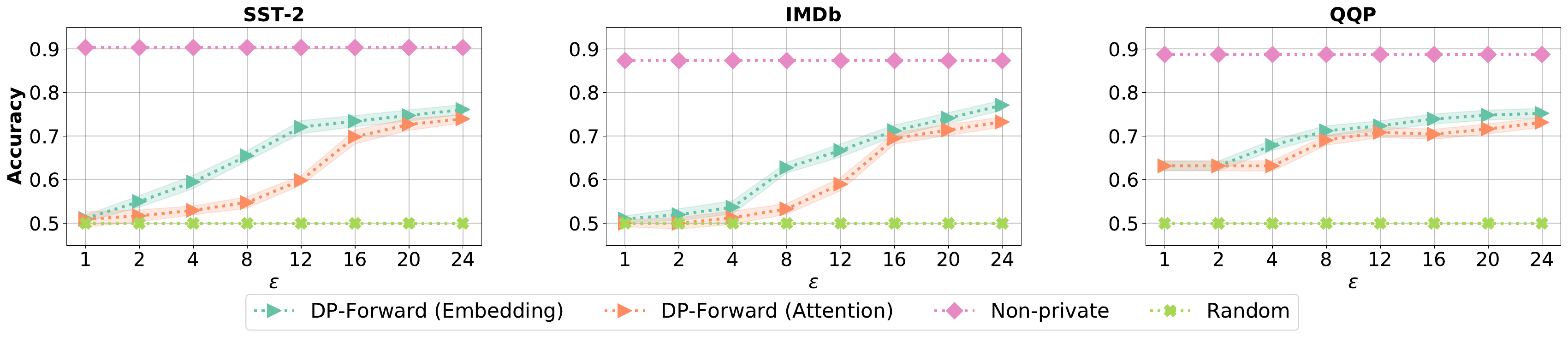}
    \vspace{-20pt}
    \caption{Privacy-accuracy tradeoff of token-level SeqLDP instances when tuning local $\epsilon$}
    \label{fig:eps_results}
    \vspace{-8pt}
\end{figure*}

\subsection{More Experiment Results}
\label{apdx:exp}
We also study the privacy-accuracy tradeoff on all three tasks for our two token-level SeqLDP designs when tuning local $\epsilon$.
The results are compared with the non-private baseline and fine-tuning using MVG noise.
Figure~\ref{fig:eps_results} shows task accuracy increases with $\epsilon$.
Perturbing input embeddings for token-level (vs. sequence-level) SeqLDP can achieve remarkable accuracy gain, \eg, ${\sim}0.7$ vs. $0.5$ for IMDb.

We evaluate the two MIAs on \mbox{SST-2} fine-tuned by our two token-level SeqLDP instances.
Table~\ref{tbl:token_MIA} shows the results, with success rates within $0.48{-}0.52$ (like random guessing) bolded.
Even if the~provable guarantee is at the token level, our instances can notably reduce the success rates of the confidence-based attack by ${\sim}14$pp and the entropy-based one by ${\sim}11$pp, compared to the non-private baseline.

\begin{table}[!t]
\centering
\begin{tabular}{c|l|cc}
\hline
\multirow{2}{*}{\begin{tabular}[c]{@{}c@{}} Local \\ $\epsilon$ \end{tabular}} & \multirow{2}{*}{Method} 
& \multicolumn{2}{c}{Attack Success Rate} \\ \cline{3-4} 
 &  & Entropy & Confidence \\ \hline
 $\infty$ & Non-private baseline 
 & $0.659$ & $0.645$ \\ \hline \hline
\multirow{2}{*}{8} & DP-Forward (Embedding) 
& $0.536$ & $\mathbf{0.503}$ \\
 & DP-Forward (Attention) 
 & $0.545$ & $\mathbf{0.516}$ \\ \hline
\multirow{2}{*}{16} & DP-Forward (Embedding) 
& $0.542$ & $\mathbf{0.509}$ \\
 & DP-Forward (Attention) 
& $0.552$ & $\mathbf{0.519}$ \\ \hline
\multirow{2}{*}{24} & DP-Forward (Embedding) 
& $0.552$ & $\mathbf{0.516}$ \\ 
 & DP-Forward (Attention) 
 & $0.559$ & $0.523$ \\ \hline
\end{tabular}
\vspace{5pt}
\caption{Success Rates of the two (sequence-level) MIAs on our token-level SeqLDP instances}
\vspace{-8pt}
\label{tbl:token_MIA}
\end{table}

\section{Relevant Matrix Algebra}
\label{apdx:matrix}

\begin{proposition}
The PDF defined in Eq.~(\ref{def:MGD}) and the matrix-trace-based one used in MVG~\cite{ccs/ChanyaswadDPM18} are equivalent.
\end{proposition}
\begin{proof}
For the numerator part in Eq.~(\ref{def:MGD}), we have
\begin{align*}
    & ~||U^{-1}(Z-M)V^{-\top}||^2_F \\
    = & ~\Tr[V^{-1}(Z-M)^\top U^{-\top}U^{-1}(Z-M)V^{-\top}] \\
    = & ~\Tr[V^{-1}(Z-M)^\top\Sigma^{-1}(Z-M)V^{-\top}],
\end{align*}
where $\Tr(\cdot)$ denotes the matrix trace.
Denote 
$$A = V^{-1}(Z-M)^\top\Sigma^{-1}(Z-M)V^{-\top}.$$
We compute
    $$B = V^{-\top}AV^\top = \Psi^{-1}(Z-M)^\top\Sigma^{-1}(Z-M),$$
which is a similar matrix of $A$,
and hence $\Tr(A) = \Tr(B)$.
So, the two PDFs are equivalent since
$$ ||U^{-1}(Z-M)V^{-\top}||^2_F = \Tr(B).\eqno\mbox{\qedhere}$$
\end{proof}

\begin{theorem}[Singular Value Decomposition or SVD~\cite{cu/HJ2012}] 
\label{theo:SVD}
A matrix $A\in \mathbb{R}^{n \times d}$ can be decomposed as $W_1 \Lambda W_2^\top$, where $W_1 \in \mathbb{R}^{n \times n} $ and $ W_2 \in \mathbb{R}^{d \times d}$ are unitary, and $\Lambda$ is an $n \times d$ diagonal matrix whose diagonal entries are ordered singular values of $A$, denoted by $\sigma_1(A) \geq \ldots \geq \sigma_r(A) \geq 0$ (or simply $\sigma(A)$) with $r = \min \{n,d\}$.
\end{theorem}

\begin{lemma}
\label{lemma_1}
Given a matrix $A \in \mathbb{R}^{n \times d}$ and two orthogonal matrices $W_1 \in \mathbb{R}^{n \times n}, W_2 \in \mathbb{R}^{d \times d}$, we have $||A||_F = ||W_1 A||_F = ||A W_2||_F$; $||\cdot||_F$ is immune to the pre- and post-orthogonal transformation.
\end{lemma}
\begin{proof}
We first prove that $||A||_F = ||W_1 A||_F$ by
$$||W_1 A||^2_F = \Tr(A^\top W_1^\top W_1 A) = \Tr(A^T A) = ||A||^2_F,$$ and similarly we can prove that $||A||_F = ||A W_2||_F$.
\end{proof}

\begin{lemma}
\label{lemma_2}
For $A \in \mathbb{R}^{n \times d}$, $||A||^2_F = \sum^r_{i=1} \sigma^2_i(A)$, where $\sigma_i(A)$ is the $i^{\text{th}}$ singular value of $A$ and $r=\min\{n,d\}$.
\end{lemma}
\begin{proof}
The SVD of $A$ is $W_1 \Lambda W_2^\top$.
By Lemma~\ref{lemma_1}, we have
$$||A||^2_F = ||W_1 \Lambda W_2^\top||^2_F = ||\Lambda||^2_F = \sum^r_{i=1} \sigma^2_i(A).
\eqno \mbox{\qedhere}
$$\end{proof}

\begin{lemma}[Lemma~$4$~\cite{corr/YangXLLW21}]
\label{lemma_3}
Given matrices $A \in \mathbb{R}^{n \times n}, B \in \mathbb{R}^{n \times d}, C \in \mathbb{R}^{d \times d}$, we have $||ABC||^2_F \leq \sum^r_{i=1}\sigma^2_i(A)\sigma^2_i(B)\sigma^2_i(C)$ where
$\sigma_i(\cdot)$ is the $i^{\text{th}}$ singular value, and $r=\min\{n,d\}$.
\end{lemma}
\label{apdx_sect3}

\begin{proof}
With SVD, $A,B,C$ are decomposed as 
$$A = W_{A_1} \Lambda_A W^\top_{A_2}, \ B = W_{B_1} \Lambda_B W^\top_{B_2}, \ C=W_{C_1} \Lambda_C W^\top_{C_2}.$$
Based on Lemma~\ref{lemma_2}, we have 
\begin{align*}
||ABC||^2_F & = ||W_{A_1} \Lambda_A W^\top_{A_2} W_{B_1} \Lambda_B W^\top_{B_2} W_{C_1} \Lambda_C W^\top_{C_2} ||^2_F \\
 & = ||\Lambda_A W \Lambda_B W' \Lambda_C||^2_F,
\end{align*}
where $W = W^\top_{A_2} W_{B_1} = (w_{ij})_{n \times n}$ and $W' = W^\top_{B_2} W_{C_1} = (w'_{ij})_{d \times d}$ are still two unitary matrices.
We further have 
\begin{align*}
    ||ABC||^2_F & = \sum^n_{i=1}\sum^d_{j=1}\sigma^2_i(A)\sigma^2_j(C)[\sum^r_{k=1}\sigma_k(B)w_{ik}w'_{kj}]^2 \\
    & = \sum^n_{i=1}\sum^d_{j=1}\sigma^2_i(A)\sigma^2_j(C) \beta^2_{ij},
\end{align*}
where $\beta_{ij}= \sum^r_{k=1}\sigma_k(B)w_{ik}w'_{kj}$.
Hence, we need to show 
\begin{align}
\label{lemma3_proof}
    \sum^n_{i=1}\sum^d_{j=1}\sigma^2_i(A)\sigma^2_j(C) \beta^2_{ij} \leq \sum^r_{i=1}\sigma^2_i(A)\sigma^2_i(B)\sigma^2_i(C).
\end{align}
Following the strategy in~\cite{corr/YangXLLW21} (\emph{cf.} Eq.~(29), (30)), we rewrite $\sigma^2_i(A)$ and $\sigma^2_j(C)$ using non-negative values $\xi_t$ and $\eta_s$ s.t.
$$\sigma^2_i(A) = \sum^n_{t=i}\xi_t, \ t \in [1,n]; \ \sigma^2_j(C) = \sum^d_{s=j}\eta_s, \ s \in [1, d].$$
For $i \in [1, n], j \in [1,d]$, we denote $\gamma_{ij} = \sigma_i(B)$, if $i=j$; $\gamma_{ij} = 0$, otherwise.
Then, we transform the Eq.~(\ref{lemma3_proof}) as 
\begin{align*}
   & \sum^r_{i=1}\sigma^2_i(A)\sigma^2_i(B)\sigma^2_i(C) - \sum^n_{i=1}\sum^d_{j=1}\sigma^2_i(A)\sigma^2_j(C) \beta^2_{ij} \\
   = & \sum^n_{i=1}\sum^d_{j=1}(\gamma^2_{ij}-\beta^2_{ij})\sigma^2_i(A)\sigma^2_j(C) \\
   = & \sum^n_{i=1}\sum^d_{j=1}(\gamma^2_{ij}-\beta^2_{ij})\sum^n_{t=i}\xi_t \sum^d_{s=j}\eta_s \\
   = & \sum^n_{t=1}\sum^d_{s=1} \xi_t \eta_s \sum^t_{i=1}\sum^s_{j=1}(\gamma^2_{ij}-\beta^2_{ij}).
\end{align*}
Since $\xi_t ,\eta_s$ are non-negative, we only need to show 
\begin{align}
\label{lemma3_eq2}
    \sum^t_{i=1}\sum^s_{j=1}(\gamma^2_{ij}-\beta^2_{ij}) \geq 0.
\end{align}
However, the original proof~\cite{corr/YangXLLW21} has two issues: i) $t > s$ is not considered, and ii) the commutative law of matrix multiplication in Eq.~(35) does not hold as $E(t)$ in Eq.~(34) is not a standard diagonal matrix.
To address them, we have 
$$\sum^t_{i=1}\sum^s_{j=1}\gamma^2_{ij} = \sum^{\min\{t,s\}}_{k=1}\sigma^2_k(B).$$
We then denote a sub-matrix $B^*=(\beta_{ij})$ for $i\in [1,t], j \in [1,s]$ of $W\Lambda_BW'$.
With SVD of $B^*$, we have 
$$\sum^t_{i=1}\sum^s_{j=1}\beta^2_{ij} = ||B^*||^2_F = \sum^{\min\{t,s\}}_{k=1}\sigma^2_k(B^*) \leq \sum^{\min\{t,s\}}_{k=1}\sigma^2_k(B).$$
The last inequality is due to $\sigma_k(B^*) \leq \sigma_k(B)$ for $\forall k \in [1, r]$~\cite{cu/HJ2012}.
So, Inequality~(\ref{lemma3_eq2}) holds.
\end{proof}

\section{Proofs for Our Analytic Matrix Gaussian Mechanism}
\label{apdx_sectB}
This section proof Lemma~\ref{lemma_4}, Lemma~\ref{lemma_5}, and Theorem~\ref{theo:dp_amgm} in Section~\ref{sect:agm-aMGM}.

\begin{proof}[Proof of Lemma~\ref{lemma_4}]
Recall that $\M(f(\D))=f(\D)+Z$ with $Z \sim \MN_{n,d}(0,\Sigma,\Psi)$, the probability of $\M(f(\D)) = O$ is 
$$\Pr[\M(f(\D) = O] =\frac{\exp(-\frac{1}{2}||U^{-1}(O -f(\D))V^{-\top}||^2_F)}{(2\pi)^{nd/2}|\Psi|^{d/2}|\Sigma|^{n/2}}.$$
Similarly, we can compute $\Pr[\M(f(\D'))= O]$.
By plugging them into $\L_{\M,\D,\D'}(O)$, and let $\Delta = f(\D)-f(\D')$,
\begin{align*}
    & \L_{\M,\D,\D'}(O) = \ln \frac{\exp(-\frac{1}{2}||U^{-1}(O-f(\D))V^{-\top}||^2_F)}{\exp(-\frac{1}{2}||U^{-1}(O-f(\D'))V^{-\top}||^2_F)} \\
    = & ~\frac{1}{2}||U^{-1}(Z+\Delta)V^{-\top}||^2_F -\frac{1}{2}||U^{-1} Z V^{-\top}||^2_F \\
    = & ~\frac{1}{2}||U^{-1} \Delta V^{-\top}||^2_F + \langle \vect(U^{-1} \Delta V^{-\top}), \vect(U^{-1} Z V^{-\top}) \rangle,
\end{align*}
where $\vect(\cdot)$ is the vectorization of a matrix and $\langle \cdot, \cdot \rangle$ denotes the inner product.
For easy presentation, we denote $Z' = U^{-1} Z V^{-\top}$ and $\Delta' = U^{-1} \Delta V^{-\top}$, and then we re-write 
$$\L_{\M,\D,\D'}(O) = \frac{1}{2}||\Delta'||^2_F+\langle \vect(\Delta'), \vect(Z')\rangle.$$
Given Lemma~\ref{lemma_6}, $Z' \sim \MN_{n,d}(0, I_n, I_d)$ with each entry i.i.d.~drawn from $\N(0,1)$.
$\langle \vect(\Delta'), \vect(Z')\rangle$ 
is thus the $\Delta'$-weighted
sum of $nd$ {i.i.d.} Gaussian random variables,
which is a Gaussian variable\footnote{\href{https://en.wikipedia.org/wiki/Sum_of_normally_distributed_random_variables}{en.wikipedia.org/wiki/Sum\_of\_normally\_distributed\_random\_variables}} $\N(0,||\Delta'||^2_F)$ too.
So $\L_{\M,\D,\D'} \sim \N(\eta,2\eta)$, $\eta = \frac{1}{2}||\Delta'||^2_F$.
\end{proof}

\begin{proof}[Proof of Lemma~\ref{lemma_5}]
With Lemma~\ref{lemma_4} and the CDF, we have
\begin{align*}
    & ~\Pr[\L_{\M,\D,\D'} \geq \epsilon] = \Pr[\N(\eta, 2\eta) \geq \epsilon] \\
    = & ~\Pr[\N(0,1) \geq \frac{\epsilon-\eta}{\sqrt{2\eta}}]
    = \Pr[\N(0,1) \leq \frac{\eta - \epsilon}{\sqrt{2\eta}}] \\
    = & ~\Phi(\frac{||\Delta'||_F}{2} - \frac{\epsilon}{||\Delta'||_F}),
\end{align*}
where we used $\N(\eta, 2\eta) = \eta + \N(0,1)/\sqrt{2\eta}$ and the symmetry of the standard normal distribution $\Pr[\N(0,1)\geq t] = \Pr[\N(0,1)\leq -t]$.

A similar argument applied to $\L_{\M,\D',\D}$ yields
$$\Pr[\L_{\M,\D',\D} \leq -\epsilon] = \Phi(-\frac{||\Delta'||_F}{2} - \frac{\epsilon}{||\Delta'||_F}).\eqno \mbox{\qedhere}
$$
\end{proof}

\begin{proof}[Proof of Theorem~\ref{theo:dp_amgm}]
The proof boils down to two directions.
From $(\epsilon,\delta)$-DP (Theorem~\ref{suff_and_nece}) to Theorem~\ref{theo:dp_amgm}, the proof directly follows from all the derivations in Section~\ref{sect:agm-aMGM}.
For the inverse direction, it is sufficient to show that $||\Delta'||_F \leq \B$ holds for $\forall \D \simeq \D'$ given Theorem~\ref{theo:dp_amgm}.
In particular, for $||\Delta'||_F = ||U^{-1}\Delta V^{-\top}||_F$, we have 
\begin{align*}
    & ~||U^{-1}\Delta V^{-\top}||_F^2
    \leq \sum^r_{i=1}\frac{\sigma^2_i(\Delta)}{\sigma^2_{n-i+1}(U)\sigma^2_{d-i+1}(V)} \\
    \leq & ~\frac{\sum^r_{i=1} \sigma^2_i(\Delta)}{\sigma^2_{n}(U)\sigma^2_{d}(V)} 
    \leq \frac{||\Delta||^2_F}{S^2_2(f)/\B^2} \leq \B^2,
\end{align*}
where the first inequality is due to Lemma~\ref{lemma_3}, the second one holds since $\sigma_n(U)$ and $\sigma_d(V)$ are the smallest singular values among the others, and the third one is from Theorem~\ref{theo:dp_amgm}.
\end{proof}

\fi

\end{document}